\documentclass[12pt,english]{article}
\usepackage{lmodern}
\usepackage[T1]{fontenc}
\usepackage[latin9]{inputenc}
\usepackage{geometry}
\usepackage{color}
\usepackage{babel}
\usepackage{amsmath}
\usepackage{amsthm}
\usepackage{amssymb}
\usepackage{graphicx}
\usepackage{color,soul}
\usepackage{setspace}
\usepackage[authoryear]{natbib}
\usepackage[unicode=true,pdfusetitle,
  bookmarks=true,bookmarksnumbered=false,bookmarksopen=false,
  breaklinks=false,pdfborder={0 0 1},backref=false,colorlinks=true]
           {hyperref} 
\hypersetup{pdfborderstyle=,citecolor=blue,linkcolor=blue,urlcolor=blue}
\makeatletter
\definecolor{green}{RGB}{0, 140, 50}

\providecommand{\tabularnewline}{\\}

\theoremstyle{remark}
\newtheorem{rem}{\protect\remarkname}
\theoremstyle{plain}

\theoremstyle{definition}
\newtheorem{defn}{\protect\definitionname}
\theoremstyle{definition}
 \newtheorem{example}{\protect\examplename}
\theoremstyle{remark}

\theoremstyle{plain}
\newtheorem{prop}{\protect\propositionname}
\theoremstyle{plain}
\newtheorem{thm}{\protect\theoremname}
\theoremstyle{plain}
\newtheorem{lem}{\protect\lemmaname}
\theoremstyle{remark}
\providecommand{\claimname}{Claim}
\providecommand{\definitionname}{Definition}
\providecommand{\examplename}{Example}
\providecommand{\factname}{Fact}
\providecommand{\lemmaname}{Lemma}
\providecommand{\propositionname}{Proposition}
\providecommand{\remarkname}{Remark}
\providecommand{\theoremname}{Theorem}

\newcommand{\shortv}[1]{}
\newcommand{\fullv}[1]{#1}
\newcommand{\commentout}[1]{}

\begin{document}
% joe10:  I'm using \shortv and \fullv in various places.  This
%automatically takes care of switching from the eEC version to the
%full version
\shortv{\setcounter{page}{0}}
\title{The Benefits of Coarse Preferences}

\shortv{\author{Submission 143}}
\fullv{\author{Joseph Y. Halpern\thanks{
Computer Science Department, Cornell University. email:
\protect\href{mailto:halpern@cs.cornell.edu }{halpern@cs.cornell.edu
}.} 
\and Yuval Heller\thanks{Department of Economics, Bar-Ilan
University and University of California San Diego. email:
\protect\href{mailto:yuval.heller@biu.ac.il}{yuval.heller@biu.ac.il}. } 
\and Eyal Winter\thanks{Management School, Lancaster University, and
Department of Economics, The Hebrew University. email:
\protect\href{mailto:mseyal@mscc.huji.ac.il}{mseyal@mscc.huji.ac.il}. } 
\thanks{
%Y(9.6): Added thanks here:
{We thank Nick Netzer, the anonymous reviewers, and the editor for various helpful comments.} 
%Y(9.6): Do we want to add any others / seminar audiences?
YH gratefully acknowledges the financial support of the European Research
Council (\#677057), the Israeli Science Foundation (\#2443/19), and
the Binational Israel-US Science Foundation (\#2020022).
{
JH was supported in party by 
NSF grants IIS-178108 and IIS-1703846 and MURI grant
W911NF-19-1-0217.}
  }}
  }

% joe10: added \fullv, but againl you don't have to do anything for the
% journal version
\fullv{  \maketitle}

\begin{abstract}
We study the strategic advantages of coarsening one's utility
      by clustering payoffs together (i.e., 
        classifying them the same way). Our solution
        concept, \emph{coarse-utility equilibrium (CUE)} requires that (1) each
        player maximizes her coarse utility, given the opponent's
        strategy, and (2) the classifications form best replies to
        one another. We characterize CUEs in various games. In
        particular, we show that there is a qualitative difference
        between CUEs 
       in which only one of the players clusters 
       payoffs and those in which all players cluster their payoffs, and that, in the latter type of CUE, players treat 
       other players better than they do in Nash
        equilibria in games with monotone
                externalities. 
\end{abstract}
\fullv{\noindent \textbf{Keywords}: 
Categorization, language, indirect evolutionary approach, monotone
externalities, strategic complements, strategic
substitutes. \ \ \ \textbf{JEL codes}: C73, %(EVOLUTIONARY GAMES)
D83 %(LEARNING/UNAWARENESS)
}

%\begin{document}

% joe10: added \shortv
\shortv{

  \begin{titlepage}

\maketitle

\end{titlepage}
}

%\newpage
\section{Introduction}
%Y(7.6): I have copied-pasted here Eyal's revised Introduction

{A \emph{strategic commitment} 
is a situation in which a player promises to restrict her
  actions in a way that, if believed by co-players, yields strategic
  benefits. 
Commitment plays an important role in a variety of strategic
environments, including negotiations, voting, and international conflicts,
and it has been studied extensively in the game theory
literature  
since the seminal work of \cite{schelling1980estrategy}.}

{
While commitment always involves certain self-imposed restrictions,
it may appear in different forms. A player can commit to taking
an action or commit or to avoid taking an action.
Such commitments operate directly on the set of
strategies available to players, and the device that makes them
credible typically involves a pre-play action that makes a
deviation from the commitment extremely costly. 
However, some commitments can be indirect and take the form of making
certain changes to preferences or certain beliefs observable, or even
making moral 
sentiments or behavioral biases observable
(see, e.g.,} 
\citealp{GuethYaari1992Explaining,DekelElyEtAl2007Evolution,heifetz2007maximize,
  winter2017mental, heller2020biased}). 
  { For such commitments, the device that guarantees
   credibility is very often cultural and hence quite often less than
   perfect. Social norms often play a role in sustaining their
   credibility.}   
 
{
In this paper, we highlight one such form of commitment and investigate
  its implications on equilibrium behavior in various strategic
  environments. Specifically, we are interested in commitments made by
coarsening preferences. By coarsening,
players form categories that are potentially less refined than their
original (material) preferences and regard any two utility levels
belonging to the same category as identical. } 

{
Our motivation in addressing the role of coarsening in games is
twofold: First, we believe that coarsening does take place in many
games, as it makes the play simpler and makes decisions easier to
explain both to oneself and to others.\footnote{{Some papers
  study categorization and 
    coarsening that results from making optimal decisions under complexity costs
or that arises from an upper bound on the number of categories 
(see, e.g. \citealp{netzer2009evolution, halpern2015algorithmic, robson2023adaptive}). By contrast, 
we abstract away from complexity costs for two reasons: First, we
are interested in coarsenings that are not driven by complexity
considerations but by
the benefits of commitment. For example. at a car dealership a buyer
may attempt to commit to treating the price of \$20,000 for his desired car
as being just as bad for him as \$22,000 and treating both prices as
strictly worse 
than \$16,000, not because it is too complex for him to distinguish
between the first two prices, but because he wishes to induce an offer
around the third price. Second, even without the cost of complexity
the model is quite rich. Adding complexity costs to the model  is
likely to make it intractable unless strong assumptions are made.}} 
This benefit of coarsening also helps make it credible as a
commitment device.  
Second, we view coarsening as an attractive example for the idea
that in real life language and actions are intimately interwoven, and 
cannot be easily separated. This idea was first proposed by
\citeauthor{wittgenstein2010philosophical} (1953, 2003)  
who use the term ``Language Game'' to refer to this interweaving of language and action and to the fact that words and sentences receive their meaning only through the context and actions in which they are spoken.}

{
In line with Wittgenstein, we view commitment
  in games as a phase of pre-play communication where language functions
in a dual role: 
\begin{enumerate}
\item as a tool that assists players in reasoning about the game and
  making strategic decisions, and 
  \item as a device by which players exchange signals regarding
            preferences, beliefs, and actions. 
\end{enumerate}  
In our specific context of coarsening preferences, language can take
the form of adjectives representing one's subjective attitude towards
game outcomes, and clustering them into terms such as
``unacceptable'', ``fair'', and ``very generous''.   A player who claims
that she needs to attain a minimal level of payoff from the game in
order to achieve a certain goal (like paying back a mortgage) is also
describing an implicit coarsening of preferences. } 

{
For commitments to become credible through pre-play deliberation it
is not necessary that players will believe the literal statements made
by their counterparts. The process of social learning may support
situations in which a player says 
%Y: I have slightly rephrased and slightly changed the notation the follwoing end of the sentence 
statement $p$ regarding his clustering and both
players understand that he means $q$.
We note here that we also
consider environments in which  
%Y(7.6): revised this end of sentnence
the observation of the opponent's coarsening is not perfect but probabilistic.}

{
To study the role of preference coarsening in games we introduce an
equilibrium concept that we call coarse utility equilibrium (CUE) for
strategic games, where the players' choice of classification is
derived endogenously, as part of the equilibrium. Specifically, we
consider a two-stage game where in the first stage
players choose their coarsening commitment, and in the second stage,
after observing the coarsening of their counterpart, they play the
underlying game.  
Roughly speaking, our solution concept can be viewed as a subgame-perfect equilibrium of this two-stage game. 
This framework allows us to study not only the equilibria that emerge
in the game but also the optimal coarsening strategies that players
use as commitments.   
 }

We describe {now in more detail} the nature of this solution concept:
Rather than considering the coarsening of arbitrary sets of outcomes,
we consider coarsening only the set of payoffs
(utility levels) of the players. Thus, coarsening results in a player bundling
payoffs into equivalence classes, 
and treating different payoffs that belong to the same class as if they
were identical. One way of {interpreting}  this process is that all
outcomes in an equivalence class are framed in the same way, and the
player\textquoteright s utility {(and hence actions)} depends only on the framing.
This process is quite natural in many games, {and captures the interweaving between language and actions that we discussed earlier.} Since utilities are real numbers, by coarsening
utilities, we can require that if two utilities $u_1$ and $u_2$ are 
in the same category, then so are all utilities in the interval
$[u_1,u_2]$.  This added structure on coarsening plays a key role in
our results. 

Our equilibrium conditions impose two best-response requirements,
one for each stage of the two-stage game. The first requires that,
based on the clustering used
in the first
stage of the game, players play a Nash equilibrium in the (strategic-form)
game of the second stage
(where players are indifferent between any two outcome 
that belong to the same cluster). The second
requires that clusterings are
optimal with respect to players\textquoteright{} unclustered (material)
payoffs, that is, that players best respond
to one another, so that their strategies form
an equilibrium with respect to the anticipated behavior 
in the second stage.
Hence, one can think of the first stage as a game where a
player\textquoteright s 
strategy is a commitment to a certain 
classification of outcomes (i.e., clustering), and its payoffs
are determined by the equilibrium outcome that arises from such commitments.

{
We interpret these two stages very differently. 
The second stage is viewed as a standard strategic game. In contrast,
the optimization in the first stage is {viewed as the
  outcome of a longer process in which culture and experience allow
  players to experiment} with different  
classification choices, {learn} how others respond to them, and {gradually} optimize their classification 
(as discussed at the end of Section \ref{sec-model}).
We do not provide a formal model for this process or the process by
which one player's choice of clustering becomes known to her
counterpart. Instead, we treat these processes as a ``black box''.
A full description of these processes would add little to the overall
insight, and would make the analysis more cumbersome than necessary.
However, we do analyze the robustness of our results in a setting where
the knowledge about the clustering choice of one's counterpart
is not perfect, but probabilistic.\footnote{ 
Specifically, in Appendix \ref{sec-partial}, {we
  extend our model to a setup with
    extend our model to a setting with 
  partial observability in which each player observes her opponent's
    clustering with probability $p<1$. We show that all of our results
    hold for a sufficiently high probability  $p<1$ of observing a
    deviation of the opponent to a different clustering. By contrast,
    if the observation probability is close to zero, then clustering
        can induce only Nash equilibria. }}}  

As an illustration, consider the prisoner's dilemma game described
  in Table \ref{tab:Prisoner-Dilemma}. The fact that  
defection is the dominant
action implies that mutual defection is the unique Nash
equilibrium. Consider a pair of players who are both committed to
bundling together all payoffs above 10.   
These players can be viewed as \emph{satisficing}, in the sense of
\citet{simon1955behavioral}, with an
aspiration level of 10. Another interpretation is that this coarse utility might be a commitment to a rounding heuristic: when a player has a two-digit payoff, he pays attention only to the leading digit. Observe that mutual cooperation is an equilibrium given these coarse utilities. 
Further observe that
a player cannot gain a higher material payoff by deviating to a different clustering because the fact that her opponent either obtains a payoff of 
at least 10 or plays
her material best reply limits the deviator's material payoff to being at
most 10. This implies that the mutual cooperation supported by
satisficing is a  
coarse-utility equilibrium.\footnote{The 
meta-analysis of \cite{mengel2018risk} shows that
there is a substantial  rate of cooperation ($37\%$) in experiments
of the (one-shot) prisoner's dilemma.}

\begin{table}
\caption{\label{tab:Prisoner-Dilemma}Matrix Payoffs of a Prisoner's Dilemma
Game (Illustrative Example)}

\centering{}%
\begin{tabular}{|c|c|c|}
\hline 
 & \textcolor{red}{\emph{c}} & \textcolor{red}{\emph{d}}\tabularnewline
\hline 
\textcolor{blue}{\emph{~~c~~}} & \emph{\Large{}$_{\underset{\,}{{\color{blue}10}}}\,^{\overset{\,}{{\color{red}10}}}$} & \emph{\Large{}${\color{blue}_{\underset{\,}{0}}}\,\,^{{\color{red}\overset{\,}{11}}}$}\tabularnewline
\hline 
\textcolor{blue}{\emph{d}} & \emph{\Large{}$_{\underset{\,}{{\color{blue}11}}}\,\,^{\overset{\,}{{\color{red}0}}}$} & \emph{\Large{}$_{{\color{blue}1}}\,\,\,^{{\color{red}1}}$}\tabularnewline
\hline 
\end{tabular}
\end{table}

We view the formation of players' classification as a 
long-term process of
cultural evolution in which players preferences dynamically change in a way that increases their social fitness, where this fitness is represented by 
the original payoffs of the game. 
We do not provide an explicit model of this process.
Instead, we treat it as a ``black box'' that
delivers an equilibrium outcome with each player's classification best
responding to that of his/her counterpart.  

In real life, such classifications
are often determined by a player's language, which both helps the
player in reasoning about the situation and in providing signals about
other players' coarsenings.  
Hence, we can interpret players'  classification as a 
dictionary, that is, a mapping from outcomes to
words.
In addition to helping players reason about the underlying situation, language also provides them with
a tool for making clustering commitments to
the other players.
This induced language can be very limited, leading to
a coarse classification with only two components (e.g., high/low or
acceptable/unacceptable, as in our earlier example),  
or can result in a much more refined classification, using a larger vocabulary
(e.g., outstanding, generous, fair, satisfactory, disappointing, insulting,
or outrageous). 
We  view the first stage of the game as one where players determine how to cluster and communicate their coarsened utility function to the other players.

In defining the equilibrium conditions for players\textquoteright{}
classification strategies, we consider three
variants of CUE that differ
in the description of circumstances under which 
must best respond to a classification strategy. The first variant takes a best response by player $i$ at stage 1 to be one where there is at least one equilibrium in the second stage where player $i$ does not gain by deviating.  The second variant requires that player $i$ does not gain by deviating in all equilibria of the stage 2 game.
Finally, the third variant requires that 
 player $i$ does better, not in all equilibria of the stage 2 game,
 but only in \emph{plausible} equilibria, where an equilibrium is
   plausible if, should one of the players deviate to a new classification, 
   the players can reach the new equilibrium by a
   sequence of changes in their play such that, at each step, the
player who changes her behavior increases her coarse utility. 
We believe that the third variant is the most reasonable, and use it
for most of our results (we discuss its 
evolutionary interpretation at the end of Section \ref{sec-model}). The weaker concept is
less informative, while the stronger one does not always exist.

Our analysis involves two domains of 
two-player games.\footnote{To simplify the notation, we  focus on
two-player games. All the definitions, and many of the
  results, can be extended to $n$-player games in a straightforward
  way.} The first is finite normal-form  games (where mixed
strategies are allowed) and the second is continuum games in which the
set of strategies is an interval. We start with a few results that  
provide some preliminary insights about our proposed solution concepts. We first show that weak CUE gives rise to a folk theorem;  the set of equilibria includes all individually rational outcomes. This result also motivates our interest in the two other, more restrictive, variants of CUE.  We next consider the relationship between strong CUE and Stackelberg leadership, showing that every Pareto efficient outcome of a game that pays each player at least her Stackelberg payoff is a strong CUE. This latter condition can be interpreted as a requirement that no player can regret not doing something else, under the assumption that his opponents would have best responded to what he would have done.

In Section \ref{subsec:constant-sum}, we study two extreme classes of
games, zero-sum games and common-interest games. Our interest in these
two classes of game is motivated by the fact that the role of
commitment is limited in games belonging to these classes. Roughly, if
players' interests are in complete conflict, and one of them commits to  
treating two outcomes as if they were identical, then it must be the
case that such a commitment makes him better off.   But this means
that it must make the player who reacts to this commitment worse off,
and hence the latter should ignore the commitment and act as if it was
never made. Similarly, if players' interests are fully aligned,
commitment is superfluous; the players can easily coordinate to arrive
at the outcome that maximizes their joint utility. Our results
coincide with these intuitions. We show that in a zero-sum game, every
Nash equilibrium is a strong CUE and every weak CUE yields the minimax
value of the game; in common-interest games, the set of CUE outcomes
coincides with the set of Nash equilibria, and every strong CUE
outcome is a Pareto-efficient Nash equilibrium. 

Section \ref{sec:interval-games} provides our main results for interval games. Our analysis here concerns games with monotone externalities (involving both negative and positive externalities), a property shared by many economic applications. 
Our first result focuses on CUE in which players do not 
best reply to their unclustered utilities (and hence differ from Nash equilibria).  Theorem \ref{Thm:neccesary-conditions-monotone} shows that in such games a player's CUE strategy treats her opponent better than her best reply to her unclustered utility would do. 
The key observation is that any opponent's reaction to a small deviation of a player must be towards the opponent's best reply to her  unclustered utility.
This result is surprising, as it applies to the cases of both strategic complements and strategic substitutes, 
while the existing related literature yields  similar results only
  in games with strategic complements (as discussed in Remark
  \ref{rem-theorem-1}). Conceptually, this results implies that  
  bilateral clustering commitments aren't used as threats, but as 
 positive commitments, that is, a promise to treat co-players better
 than when such a commitment is not made. 

We also show that when both players deviate from a CUE towards a profile in which each player's externalities reduce the utility of his co-player, then not only can such a profile not Pareto dominate the CUE, but also no single player can be made better off assuming his co-player further deviates by best responding with his unclustered preferences. This feature of CUEs can be viewed as a stability property, highlighting the fact that in CUEs players utilize their externalities in a cohesive and welfare-enhancing manner. 

Next we add the assumption of the game having strategic
complementarity, and we show that the two properties we described
above,  
beneficial commitment and stability, become both necessary and
sufficient for all CUE in games with strategic complementarity, and
hence fully characterize CUEs for these games. 

We then explore CUE outcomes in the class of games with strategic
substitution. Interestingly, in these games there is a sharp
  distinction between CUEs 
  in which only one of the players clusters her utility, and CUEs 
  in which both players cluster their utilities. 
Specifically, in CUEs with only one player (say, Alice) playing the
best reply to her unclustered utility, the other player's (Bob's) CUE strategy will
treat Alice less favorably than would be the case under Bob's unclustered best
reply.  
Hence, in CUEs in which only one of the players clusters her utility, the clustering-induced commitment is best viewed as a threat rather than a favorable promise (while bilateral CUE commitments are still interpreted as favorable promises). 

An important finding of our analysis concerns the comparison of CUE outcomes between 
games with strategic complementarity and games with strategic substitution. Roughly, the first property involves common interest, while the second involves  opposing interests.  Under complementarity, in CUE outcomes, players who deviate from rationality make their co-player better off, but under substitution, players who (unilaterally) deviate from rationality  make their co-player worse off. This structure of deviation from rationality considered by our solution concept is interesting because it contrasts with other behavioral solution concepts (in particular, those that are based on the idea of altruism or spitefulness), where the deviation from rationality on the cooperation-competition scale is unidirectional. 

This feature of our solution concept also connects to some experimental findings showing that when players perceive the strategic environment as competitive,
they compete more vigorously, and when they perceive it to be cooperative, they cooperate more willingly, even when rationality prescribes the same behavior (see, e.g, \citealt{goerg2010prevalence}). It also
connects to the so-called ``Social Salience Hypothesis'' regarding the neurotransmitter Oxytocin that regulates social behavior (see, e.g., \citealt{shamay2016social}).
Early studies of this hormone raised the hypothesis that its main evolutionary role is to enhance cooperation. But as evidence started piling up showing that in certain strategic environments Oxytocin makes people more competitive/aggressive, a new hypothesis has emerged claiming that Oxytocin boosts competition in competitive environments and enhances cooperation in cooperative environments. Our finding might contribute to this hypothesis by suggesting the possibility that the evolutionary forces that sustained this dual role of the hormone benefited from the strategic commitment that such behavior generates. 

%Y(14.6): I tentatively removed this paragraph
%To demonstrate our solution concept in a more applied framework, we provide two examples. One concerns a model of price competition with differentiated goods (satisfying strategic complementarity) and the other involves a Cournot model (satisfying strategic substitutability).

The rest of the paper is structured as follows.
In the remainder of this section, we briefly survey the related
literature. Section 
\ref{sec-model} 
presents our model and solution concept. In Section \ref{sec-results},
we present our general results. Section \ref{sec:interval-games}
characterizes CUE in Interval games with monotone externalities. 
The appendix contains formal proofs and an extension of our
  model to deal with partial observability of the opponent's clustering.

\subsection{Related Literature}\label{subsec-relate}
Our paper is related and inspired by three strands of literature. The
first strand includes papers that study the impact of
categorization. A few papers have studied categorization in
single-agent decision problems. \citet{Mull02} studies an agent who is
constrained to choosing a category, rather than a more refined choice,
and examines the kinds of biases that arise as a result of categorical
thinking. \citet{mengel2012evolution} compares  the evolutionary
fitness of different categorization of decision
situations. \citet{mohlin2014optimal}  
studies optimal categorizations that minimize the prediction error. 
\citet{horan2019coarseness} characterize conditions under
which having coarse utility can benefit a decision maker who perceives the values of alternatives with noise. An
important difference between our paper and those cited above is that
we study categorization in multi-player games; the strategic
implications of players' categorizations plays an important role
in our solution concept. 

Other papers have considered categorization in multi-player strategic
interactions. 
\citet{jehiel2005analogy} and \citet{jehiel2008revisiting} consider
multi-stage games,  
where each player $i$ bundles nodes in the game tree in which other
players move into what they call \emph{analogy classes}.  Player 
$i$ assumes that player $j$ makes the same move at all nodes in the
same analogy class. 
\citet{azrieli2009categorizing} studies games with many players and
shows that 
categorizing the opponents into a few groups can lead to efficient
outcomes.  \citet{steiner2015price} study the price distortions that
are induced when traders apply coarse reasoning in their
forecasts. \citet{daskalova2020categorization} examine how attempts to
coordinate predictions with others affects incentives for coarse 
categorization in different environments. A key difference between
these models and ours is that in the models in the other papers that
we mentioned, the categorization is determined
exogenously, whereas in our model, the categorization are endogenously 
determined as part of the solution concept.

%Y(24.4): A revised and much-extended paragraph from here due to R1's request to cite 2 more papers, and to move Gauer2020 from afootnote to main text
{A few related papers study situations in which players face different games, and have to categorize the games and decide how to play in each class.} \citet{mengel2012learning} {studies how players jointly learn how to bundle different games and how to adjust their behavior in each class. This learning model is further developed, and experimentally tested, in} \citet{licalzi2022feature}. 
\citet{heller2016rule}
{
studied a related solution concept in which agents who interact in
various games endogenously bundle different games together, where this bundling has a commitment advantage.} \citet{gauer2020cognitive} 
{show that player involved in a class of conflict games may prefer having coarse information about the opponent's utility even when the (positive)  cost of accurate information is arbitrarily small. One key point in which our model differs from these papers is that in our setup player face a single game, and they bundle together some intervals of their payoff function (rather than bundling together different games).}

The second strand of literature includes papers that study the stability of endogenous preferences 
 (the indirect evolutionary approach; see, e.g.,
\citealp{GuethYaari1992Explaining,DekelElyEtAl2007Evolution,heifetz2007maximize,friedman2009equilibrium,Herold_Kuzmics_2009,eswaran2014economic,winter2017mental}), 
and those in which a player can choose a delegate (with different 
incentives) to play on his behalf (see, e.g.,
\citealp{fershtman1997unobserved,
  Dufwenberg_Guth_1999_EvoandDel,fershtman2001strategic}). Similar to
our model, these papers allow the player's subjective payoffs to
differ from the material payoffs, and assume that a deviation to new
subjective payoffs induces the players to move to a new equilibrium.  

Our paper differs from the papers mentioned above in that we substantially
restrict how much the subjective 
(clustered) utility is allowed to differ from
the material (unclustered) utility. The only 
difference we  allow is that of clustering together
intervals of outcomes. Intuitively, these are outcomes that
the
agent commits to regarding as identical (intuitively, ones that he
would describe the same way).  By contrast, the existing literature 
is much more permissive; it typically allows an agent to have an arbitrary
subjective utility, including one in which she prefers a bad outcome
to a better outcome. We think that our restriction is reasonable
in many setups. For example, students or teachers may cluster grades  that are
given on the 0-100 scale into coarser categories, viewing 
grades within the same category as essentially
identical (see the related analysis of the optimal coaresning of grades by
\citealp{dubey2010grading}). 
By contrast, it seems unlikely that students would strictly prefer obtaining a low grade to obtaining a higher grade.

The third strand of literature consists of work studying the role of
commitment in strategic situations. This topic has been extensively
investigated since the seminal work of \citet{schelling1980estrategy} 
(see, e.g.,
\citealp{bade2009bilateral,renou2009commitment,arieli2017sequential}). Our
contribution with respect to this literature involves introducing a new
commitment device, which seems plausible in many real-life situations:
clustering intervals of payoffs together. Our results show that such clustering can result in novel outcomes; for example, equilibria in which only one of the players clusters her payoffs are qualitatively different than those in which both players cluster their payoffs (the two classes of CUEs induce the opposite behaviors in
games with strategic substitutes, as shown in Theorems
\ref{Thm:neccesary-conditions-monotone} and \ref{thm-substitues}).

\section{Model }\label{sec-model}

\paragraph{Underlying Game}
%Y(11.1): Internal comment for ourselves: The results that hold for n-player games are: Prop. 1--4 + 8 (all of which are straightforward results). The more interesting results (Prop. 5-7 + Theorem 1-3 + all of our examples) hold only for 2-player games.
Let $G=\left(S,\pi\right)$ be a two-player
normal-form game (which
we refer to as the \emph{underlying game}), where:
\begin{enumerate}
\item $S=S_{1}\times S_{2}$ is the set of strategy profiles, where each
$S_{i}$ is a convex and compact subset of a
Euclidean space that represents the set of strategies
of player $i\in\left\{ 1,2\right\} $; 
\item $\pi=\left(\pi_{1},\pi_{2}\right)$ is the profile of unclustered (material) utilities,
where each $\pi_{i}:S\rightarrow\mathbb{R}$ is a function assigning
each player $i\in\left\{ 1,2\right\} $ a payoff for each strategy
profile. 
\end{enumerate}
We use $i\in\left\{ 1,2\right\} $ as an index referring to one of the players. Let $-i$ denote the opponent of player $i$. We assume
each payoff function $\pi_{i}\left(s_{i},s_{-i}\right)$ is continuous
in all parameters and is weakly concave in the player's own strategy
($s_{i}$).
For each two strategies $s_{i},s_{i}'\in S_{i}$ and each $\alpha\in\left(0,1\right)$,
let $\left(1-\alpha\right)\cdot s_{i}+\alpha\cdot s'_{i}\in S_{i}$
denote the strategy that is a convex combination of $s_{i}$ and $s'_{i}$.
Strategy profile $s$ is \emph{interior }if $s_{i}\in\textrm{Int}\left(S_{i}\right)$
for each player $i$.

We are particularly interested in two classes of underlying games: interval
games and finite games. We say that the game is an \emph{interval
game} if each $S_{i}$ is a bounded interval in $\mathbb{R}$ (e.g.,
each player chooses a real number representing quantity, price, or
effort). We say that the game is \emph{a finite} \emph{game} if each
$S_{i}$ is a simplex over a finite set of pure actions (i.e., $S_{i}=\Delta\left(A_{i}\right)$,
where $A_{i}$ is finite), and each $\pi_{i}$ is a von Neumann--Morgenstern
payoff function (i.e., it is linear with respect to the mixing probabilities). 

With a slight abuse of notation we identify a pure action $a_{i}$
with the degenerate strategy that assigns probability one to the action
$a_{i}$.  Note that a two-action game (in which, $\left|A_{i}\right|=2$
for each player $i$) is both a finite game and an interval game (where
we identify each strategy $s_{i}$ with the probability it assigns
to the first pure action).

Let $BR_{i}:S_{-i}\rightarrow S_{i}$ denote the (unclustered) \emph{best-reply
correspondence}, and let $BRP_{i}:S_{-i}\rightarrow\mathbb{R}$ denote
the best-reply (unclustered) payoff; that is, 
\[
BR_{i}\left(s_{-i}\right)=argmax_{s_{i}\in S_{i}}\left(\pi_{i}\left(s_{i},s{}_{-i}\right)\right),\,\,\,BRP_{i}=\max{}_{s_{i}\in S_{i}}\left(\pi_{i}\left(s_{i},s{}_{-i}\right)\right).
\]
The continuity and weak concavity
of $\pi_{i}$ implies that $BR_{i}\left(s_{-i}\right)$ is a non-empty
closed convex set, {which in turn implies that the game has a Nash
equilibrium.}  If $\pi_i$ is stictly concave, then $BR_i(s_{-i})$
is a singleton 
(in which case $BR_{i}$ is a single-valued function).

In an interval game we say that strategy profile $s'$ is lower than
strategy $s$ (denoted by $s'<s$) if $s'_{i}<s_{i}$ for each player
$i$. 
Similarly, we write $s_{i}\leq BR_{i}\left(s_{-i}\right)$ (resp., $s_{i}\geq BR_{i}\left(s_{-i}\right)$)
if $s_{i}$ is weakly lower (resp., higher) than all elements of
the set $BR_{i}\left(s_{-i}\right)$.

\paragraph{Coarse-Utility Game}

We allow players to cluster together intervals of payoffs as
equivalent
outcomes. Formally, a \emph{clustering} is a weakly increasing function
$f_{i}:\mathbb{R}\rightarrow\mathbb{R}$. The clustering $f_{i}$
describes which intervals of payoff player $i$ clusters together;
i.e., which payoffs $x\neq y$ satisfy $f_{i}\left(x\right)=f_{i}\left(y\right),$where
this latter equality implies that $f_{i}$ clusters together all payoffs
in the interval between $x$ and $y$. 

Each clustering $f_{i}$ induces a clustered utility
$f_{i}\circ\pi_{i}:S\rightarrow\mathbb{R}$ for player $i$ that
coarsens her original (unclustered)
utility $\pi_{i}$. The coarse utility $f_{i}\circ\pi_{i}$ is similar
to the unclustered utility $\pi_{i}$, except that for some intervals
of payoffs, the player clusters together all payoffs within
the interval, and subjectively considers them as equivalent payoffs.
Observe that $\pi_{i}\left(y\right)=\pi_{i}\left(x\right)\Rightarrow\left(f_{i}\circ\pi_{i}\right)\left(y\right)=\left(f_{i}\circ\pi_{i}\right)\left(x\right),$
and $\pi_{i}\left(y\right)>\pi_{i}\left(x\right)\Rightarrow\left(f_{i}\circ\pi_{i}\right)\left(y\right)\geq\left(f_{i}\circ\pi_{i}\right)\left(x\right).$ 

Observe that the only aspects of the clustering that affect the player's
preferences are the intervals of payoffs that are clustered together.
That is, if $f_{i}$ and $g_{i}$ are two clusterings with the same
clustered intervals, i.e., $f_{i}\left(x\right)=f_{i}\left(y\right)\Leftrightarrow g_{i}\left(x\right)=g_{i}\left(y\right)$,
then they both induces the same clustered preferences: i.e., $f_{i}\left(\pi_{i}\left(s\right)\right)\geq f_{i}\left(\pi_{i}\left(s'\right)\right)$
$\Leftrightarrow$ $g_{i}\left(\pi_{i}\left(s\right)\right)\geq g_{i}\left(\pi_{i}\left(s'\right)\right)$.

A specific class of clusterings that we will frequently use in the
paper are those in which a player clusters together payoffs iff they
are within a given interval. It will be useful to introduce a notation
for this frequently-used class. For each interval $\left[a,b\right]\subseteq\mathbb{R}$,
let $f_{i}^{\left[a,b\right]}$ be a coarse utility that clusters
together payoffs in the interval $\left[a,b\right]$, and does not
cluster payoffs outside $\left[a,b\right]$, i.e.,

\[
f_{i}^{\left[a,b\right]}\left(x\right):=\begin{cases}
x & x\notin\left[a,b\right]\\
a & x\in\left[a,b\right].
\end{cases}
\]
In particular, $f_{i}^{\mathbb{R}}\equiv0$ is the coarsest clustering that clusters all the payoffs
together, $f_{i}^{\emptyset}\equiv Id_{i}$ is the identity clustering
that does not cluster any payoffs together, and
$f_{i}^{\geq c}:=f_{i}^{[c,\infty)}$ 
    clusters all payoffs above $c$ together (i.e., players are satisficing,
with an aspiration level of $c$).

Given an underlying game $G=\left(S,\pi\right)$ and a clustering profile
$f=\left(f_{1},f_{2}\right)$, let the \emph{coarse-utility} \emph{game}
$G_{f}=\left(S,f\circ\pi\right)$ be the game in which the utility
of each player $i$ is $f_{i}\circ\pi_{i}$ (rather than the unclustered
utility $\pi_{i}$). Let $NE\left(G_{f}\right)$ denote all Nash equilibria
of $G_{f}$. It is easy to see that every Nash equilibrium of the
underlying game $G$ is a Nash equilibrium of the coarse-utility game
$G_{f}$. Formally:

\begin{prop}\label{fact1:any-NE-is-NE-clustered}
  $NE\left(G\right)\subseteq NE\left(G_{f}\right)$ 
for each underlying game $G$ and each profile $f$. 
\end{prop}
\begin{proof}
If   $s\in NE\left(G\right)$ then for all strategies $s_i'$ for player
$i$, we have $\pi_{i}\left(s\right)\geq\pi_{i}\left(s'_{i},s_{-i}\right)$.  Since $f_i$ is
weakly increasing, it follows that 
$f_{i}\circ\pi_{i}\left(s\right)\geq
  f_{i}\circ\pi_{i}\left(s'_{i},s_{-i}\right)$.  Thus, $s\in
  NE\left(G_{f}\right)$, as desired.
\end{proof}

\paragraph{Weak and Strong Coarse-Utility Equilibrium (CUE)}

Our solution concept is a pair consisting of a coarse-utility profile
and a strategy profile
such that: (1) each strategy is a clustered
best reply to the opponent's strategies, given the player's clustering,
and (2) each clustering is a best reply to the opponent's clustering,
in the sense that deviating to a different clustering would lead to
an equilibrium of the new coarse-utility game induced by this deviation
in which the deviator is outperformed (relative to the deviator's
unclustered payoff in the original equilibrium). 

The fact that coarse-utility games typically admit multiple equilibria
means that there are several ways to
formalize the second condition.
We consider three variants of our
basic solution concept.  
In the first variant,
weak coarse-utility equilibrium, ``best reply'' is taken to mean that the deviator 
must 
be outperformed
in at least one equilibrium in the new coarse-utility game. Formally, 
\begin{defn}
\label{def-weak-CUE} A \emph{weak CUE} is a pair $\left(f,s\right)$, where $f$ is a clustering profile
and $s$ is a strategy profile satisfying: (1) $s\in NE\left(G_f\right)$,
and (2) for each player $i$ and each clustering $f'_{i}$ , there
exists an equilibrium $s'\in NE\left(G_{\left(f'_{i},f_{-i}\right)}\right)$
such that $\pi_{i}\left(s'\right)\leq\pi_{i}\left(s\right)$. 
\end{defn}

The following example shows that the notion of weak CUE is too
permissive in the sense that it allows unreasonable behavior after
a player changes her coarse utility.
\begin{example}[Implausible weak CUE]
\label{exm-weak-is-too-weak}Consider a symmetric Cournot game with
linear demand $G=\left(S,\pi\right)$: $S_{i}=\left[0,1\right]$ and
$\pi_{i}\left(s_{i},s_{-i}\right)=s_{i}\cdot\left(1-s_{i}-s_{-i}\right)$
for each player $i$ (where $s_{i}$ is interpreted as the quantity
chosen by firm $i$, the price of both goods is determined by the
linear inverse demand function $p=1-s_{i}-s_{-i}$, and the marginal
cost of each firm is normalized to be zero). Then $((f_{1}^{\mathbb{R}},f_{2}^{\mathbb{R}}),(0.5,0))$
is a weak CUE in which both players cluster all payoffs together,
player $1$ plays $0.5$ and gains the maximal feasible payoff of
0.25 and player $2$ plays $0$ and gets zero payoff. If player 2
deviates to another clustering, then player 1 (who is indifferent
between all strategies) ``floods'' the market with quantity 1, which
yields a non-positive payoff to both players. The reaction of player
1 to the new clustering of player 2 seems implausible in the sense
that her CUE quantity $0.5$ is a clustered best reply to all of
her opponent's strategy; she has no reason to increase her quantity
to 1.
\end{example}

%Y(24.4): revised and extended paragraph +footnote  due to R2's comment 3
{
The notion of weak CUE is equivalent to a subgame-perfect equilibrium 
of a two-stage game in which in the first stage each player chooses
a clustering for her second-stage self, and in the second stage each
second-stage self chooses an action. Example} \ref{exm-weak-is-too-weak} {suggests that including the player's off-the-equilibrium path behavior after each clustering profile as part of the player's strategy is too permissive in our setup. This is so because subgame-perfect equilibrium behavior allows a player with the coarsest utility to use an  
extreme punishment that minimizes the opponent's payoff in response to any change in the opponent's utility.}\footnote{{Another reason for not defining a player's strategy to include  behavior after off-the-equilibrium-path clustering profiles is that such a strategy would be very complicated, and there is no reason to believe that a learning process would induce a large set of off-the-equilibrium path pre-specified behaviors.} }

We next define a more restrictive equilibrium notion, strong CUE,
which requires that 
a deviator who chooses a different coarse utility
is outperformed in \emph{all} equilibria of the induced coarse-utility
game. Formally: 
\begin{defn}
\label{def-strong-CUE}A\emph{ strong CUE} is a pair $\left(f,s\right)$, where $f$ is a clustering
profile and $s$ is a strategy profile satisfying: (1) $s\in NE\left(G_{f}\right)$,
and (2) $\pi_{i}\left(s'\right)\leq\pi_{i}\left(s\right)$ for each
player $i$, each clustering $f'_{i}$, and each equilibrium $s'\in NE\left(G_{\left(f'_{i},f_{-i}\right)}\right)$.
\end{defn}

Next we show that any strong CUE must Pareto dominate every Nash equilibrium of the game. This suggests that the notion of strong CUE is too restrictive, because in many games, such
as the Battle of Sexes (see Table \ref{tab:BoS-Payoff-Matrix}), no
strategy profile Pareto dominates all Nash equilibria (and, thus, many
games do not admit strong CUE). 

\begin{prop}
\label{prop:strong-necce}Let $\left(f,s\right)$ be a strong CUE,
and let $s^{NE}$ be a Nash equilibrium. Then $\pi_{i}\left(s\right)\geq\pi_{i}\left(s^{NE}\right)$
for all players $i$.
\end{prop}
\begin{rem}
  Observe that a strategy profile remains an equilibrium of a
    coarse-utility game if we coarsen the clustering of one of the
    players. Thus, if $s'\in
    NE\left(G_{\left(f'_{i},f_{-i}\right)}\right)$, then $s'\in
    NE\left(G_{\left(f^{\mathbb{R}}_{i},f_{-i}\right)}\right).$ This
    implies that we can focus in Definitions
  \ref{def-weak-CUE}--\ref{def-strong-CUE} on deviations to
    the coarsest clustering $f^{\mathbb{R}}_{i}$; that is, we can
   replace ``each clustering $f'_i$'' in Definitions
  \ref{def-weak-CUE}--\ref{def-strong-CUE}
  with ``the coarsest     clustering $f^{\mathbb{R}}_{i}$.'' 
\end{rem}
\begin{proof}
Assume to the contrary that $\pi_{i}\left(s\right)<\pi_{i}\left(s^{NE}\right)$.
Consider a deviation by player $i$ to an arbitrary clustering $f'_{i}$.
Proposition \ref{fact1:any-NE-is-NE-clustered} implies that $s^{NE}\in NE\left(G_{\left(f'_{i},f_{-i}\right)}\right)$,
which contradicts $\left(f,s\right)$ being a strong CUE. 
\end{proof}
\begin{table}
\begin{centering}
~%
\begin{tabular}{c|cc}
 & \textcolor{red}{\emph{$a_{2}$}} & \textcolor{red}{\emph{$b_{2}$}}\tabularnewline
\hline 
\textcolor{blue}{\emph{$a_{1}$}} & \textcolor{blue}{2},~~\textcolor{red}{1} & \textcolor{blue}{0},~~\textcolor{red}{0}\tabularnewline
\textcolor{blue}{\emph{$b_{1}$}} & \textcolor{blue}{0},~~\textcolor{red}{0} & \textcolor{blue}{1},~~\textcolor{red}{2}\tabularnewline
\end{tabular}
\par\end{centering}
\caption{\label{tab:BoS-Payoff-Matrix}Payoff Matrix of Battle of the Sexes
Game}
\small{
The fact that no strategy profile Pareto dominates both Nash equilibria
($\left(a_{1},a_{2}\right)$ and $\left(b_{1},b_{2}\right)$) implies
(due to Claim \ref{prop:strong-necce}) that the game does not admit
any strong CUE.}
\end{table}

\paragraph{Plausible Equilibria and CUE}

The third equilibrium notion we consider lies between weak CUE and strong CUE; we believe it is the most appropriate equilibrium notion.
To define it, we first need to define which equilibria are likely
to emerge after a player deviates to a new clustering. We assume that
the strategy profile played in the CUE is
focal, and that a player will change her behavior with respect to
this focal profile only if this change increases her clustered
payoff. An equilibrium of the new coarse-utility game is \emph{plausible}
if it can be can reached by a sequence of deviations (starting from
the focal strategy profile $s$) such that each deviation improves
the deviator's clustered payoff. 
\begin{defn}
\label{def:plausible}Fix a strategy profile $s$ and a clustering
profile $f$. An equilibrium $s'\in NE\left(G_{f}\right)$ is \emph{plausible}
with respect to $s$ if there is a sequence $\left(s^{k}\right)_{k\geq0}$
of strategy profiles satisfying: (1) $s^{0}=s$, (2) $\lim_{k\rightarrow\infty}s^{k}=s'$,
and (3) if $s_{i}^{k+1}\neq s_{i}^{k}$ for player $i$ and $k\geq0$,
then $f_{i}\left(\pi_{i}\left(s_{i}^{k+1},s_{-i}^{k}\right)\right)>f_{i}\left(\pi_{i}\left(s_{i}^{k},s_{-i}^{k}\right)\right)$.
We refer to such a sequence $\left(s^{k}\right)_{k\geq0}$ as in \emph{improvement
path}. 
\end{defn}
Let $PNE\left(G_{f},s\right)$ be the set of plausible equilibria
with respect to $s$. A CUE is a pair $\left(f,s\right)$ that satisfies:
(1) each strategy $s_{i}$ is a clustered best reply to the opponent's
strategy, and (2) each clustering $f_{i}$ is a best reply
to the opponent's clustering, in the sense that a player who
chooses a different clustering will be outperformed in each
plausible equilibrium of the new coarse-utility game. 
\begin{defn}
\label{def-CUE}A coarse-utility equilibrium\emph{ (CUE)} is a pair
$\left(f,s\right)$, where $f$ is a clustering profile and $s$ is
a strategy profile satisfying: (1) $s\in NE\left(G_{f}\right)$, and
(2) $\pi_{i}\left(s'\right)\leq\pi_{i}\left(s\right)$ for each player
$i$, each clustering $f'_{i}$, and each plausible equilibrium $s'\in PNE\left(G_{\left(f'_{i},f_{-i}\right)},s\right)$.
\end{defn}

\addtocounter{example}{-1}
\begin{example}[revisited]
 The implausible weak CUE in Example \ref{exm-weak-is-too-weak}
is not a CUE. If player 2 changes her clustering into $f_{2}^{\emptyset}\equiv I_{d},$
then the unique plausible equilibrium of the induced coarse-utility
game $G_{\left(f_{1}^{\mathbb{R}},f_{2}^{\emptyset}\right)}$ is $\left(0.5,0.25\right),$
which yields player 2 a positive payoff, contradicting the assumption
that $((f_{1}^{\mathbb{R}},f_{2}^{\mathbb{R}}),(0.5,0))$ (which results
in player 2 getting a payoff of zero) is a CUE. 
\end{example}

%Y(24.4): New sentnece due to reviewer 1.
{
Our restriction to plausible equilibria in the 
definition of CUE has a similar motivation to that of
the restriction to focal equilibria in the definition of \emph{stable configurations} given by }\citet{DekelElyEtAl2007Evolution}. 
Strategy profile $s$ is a \emph{CUE outcome} (resp.,
\emph{weak CUE outcome}, \emph{strong CUE outcome}) if there exists a clustering
profile $f$ such that $\left(f,s\right)$ is a CUE (resp., weak CUE,
strong CUE). 

The following simple observation shows that every Nash equilibrium is
a CUE outcome with respect to every clustering profile. This
implies that every game admits a CUE outcome. 
\begin{prop}
\label{fact2:any-NE-is-NE-clustered-1}Let $s^{NE}$ be a Nash equilibrium
of the underlying game. Then $\left(f,s^{NE}\right)$ is a CUE for every clustering profile $f$.
\end{prop}

\begin{proof}
The proposition holds because $s^{NE}\in PNE\left(G_{\left(f'_{i},f_{-i}\right)},s\right)$
for every player $i$ and clustering $f'_{i}$ with respect to the
constant improvement path
$\left(s^{k}\right)_{k\geq0}$=$\left(s^{NE},s^{NE},...,\right)$.  
\end{proof}
In Appendix \ref{sec-partial}, we extend our model to a
  setup with partial observability in which each player privately
  observes her opponent's clustering with probability $p\in [0,1].$ 
{Specifically, we show that all of our results hold for a sufficiently high probability  $p<1$ of observing a deviation of the opponent to a different clustering. By contrast, if the observation probability is close to zero, then clustering can only induce Nash equilibria (see Proposition \ref{prop-no-observability-only-Nash}).}
\paragraph {Evolutionary/Learning Interpretation of CUE}

We think of CUE as a reduced-form solution concept capturing
the essential features of an evolutionary process of social learning in which an agent's coarse utility determines her behavior, the induced behavior determines the agent's success, and success regulates the evolution of coarse utilities (in line with the indirect evolutionary approach, discussed in Section
\ref{subsec-relate})
In what follows, we briefly and informally
present our evolutionary interpretation. 
{
In the process, we also motivate our assumption that players' clusterings are observable.}

Consider two large populations of agents: agents who play the role of player 1 and agents who play the role of player 2. In
each round, agents from each population are randomly matched to play the underlying game against opponents from the other population. Each agent in each population is endowed with a coarse utility, and the agents play a Nash equilibrium of the game induced by their coarse utilities. For simplicity, we focus on ``homogeneous'' populations in which all agents have the same coarse utility. 

With small probability, a few agents (mutants) in one of the populations (say, population 1) may be endowed with a different coarse utility due to a random error or experimentation. We assume that agents of population 2 observe 
whether their opponents are mutants, and that the agents of population 2 and
the mutants of population 1 gradually adapt 
their play, converging to a 
(plausible) equilibrium of the new clustered game (this gradual process corresponds to an improvement path in Definition \ref{def:plausible}). 
%Y(24.4): New paragraph due to Reviewer 2's comment 3
{
The assumption of being able to observe mutants with different coarse utility can be justified either by (1) having   pre-play cheap talk, in 
which having different coarse utility induces the 
mutants to communicate differently, or (2) assuming that changes in coarse utility are much slower than changes in the players' behavior, which allows a slow process in which players learn to identify  the presence of mutants with a new coarse utility and a faster process in which they adjust their behavior when being matched with these mutants.} 
Finally, we assume that the total success (fitness) of agents is monotonically
influenced by their (unclustered) payoff in the underlying game, and that there is a
slow process in which the composition of the population evolves. 

This slow process might be the result of a slow flow of new agents who join the population.
Each new agent randomly chooses one of the incumbents in his own population
as a ``mentor'' (and mimics the mentor's coarse utility), where the probabilities are
such that agents with higher fitness are more likely to be chosen as mentors. If the
original population state is not a CUE, 
then there are mutants who
outperform the remaining incumbents in their own population, which in turn implies that the original population state is not stable, as new agents are likely to mimic
more successful mutants. By contrast, if the original population state is a CUE, 
then for any mutant there is a new
plausible equilibrium in which the mutants are weakly outperformed relative to the incumbents of 
their own population; this allows the CUE to remain stable.
{As shown in Appendix \ref{sec-partial}, the same arguments 
also work in setups with partial observability, as long as the probability of observing mutants with different coarse utility is sufficiently high.}

\section{Results}\label{sec-results}

In this section we present various results that characterize weak
CUE, CUE, and strong CUE outcomes in various classes of games.

\subsection{A Folk Theorem for Weak CUE Outcomes}

Coarse utility allows a limited form of commitment, relative to the
existing literature on the indirect evolutionary approach and on delegation,
in the sense that it allows a player to be indifferent only between
payoffs in an interval. Nevertheless, our next result shows 
that a ``folk-theorem''
result holds in our setup with respect to weak CUE outcomes. Specifically,
we show that any individually
rational profile is a weak CUE outcome. 

Before presenting the result, we formally define the  maximin (and minimax) payoff and individual rationality. The maximin (resp., minimax) value of player
$i$, $\underline{M}_{i}$ (resp., $\overline{M}_{i}$), is the highest payoff player $i$ can
guarantee herself without knowing (resp., when knowing) her opponent's strategy. Formally, 
\[
\underline{M}_{i}=\max_{s_{i}\in S_{i}}\min_{s_{-i}\in S_{-i}}\pi_{i}\left(s_{i},s_{j}\right),\,\,\,\,\overline{M}_{i}=\min_{s_{-i}\in S_{-i}}\max_{s_{i}\in S_{i}}\pi_{i}\left(s_{i},s_{j}\right).
\]
It is immediate that $\underline{M}_{i}\leq\overline{M}_{i}$. Von Neumann's Minimax theorem implies  that the
two values coincide if each player's payoff is weakly convex in the
opponent's strategy. 

A strategy profile is weakly
(strongly) individually rational if it yields
each agent a payoff weakly (strictly) higher
than her  maximin (minimax) payoff
$M_{i}$. Formally,
\begin{defn}
  A strategy profile $s\in S$ is weakly (resp., strongly) individually
  rational if 
$\pi_{i}\left(s\right)\geq \underline{M}_{i}$ (resp., $\pi_{i}\left(s\right) > \overline{M}_{i}$)
for each player $i$. 
\end{defn}
Proposition \ref{prop-falk-thm} presents a ``folk-theorem'' (i.e.,
a general feasibility theorem) for weak CUE outcomes. Specifically,
it shows that (1) any strongly individually rational profile is a
weak CUE outcome, and (2) any weak CUE outcome is individually rational.
\begin{prop}
\label{prop-falk-thm}Let $s$ be a strategy profile.
\begin{enumerate}
\item If $s$ is a weak CUE outcome, then $s$ is individually rational. 
\item If $s$ is strongly individually rational, then $s$ is a weak CUE
outcome. 
\end{enumerate}
\end{prop}

\begin{proof} {[Sketch]}
Part (1) holds because if a profile is not individually rational, a
player can deviate to not clustering any payoffs, and obtain  an
unclustered payoff of at least $\underline{M}_{i}$. Part (2) holds
because any strongly individually rational profile can be supported by
both players clustering all payoffs together, and by punishing
deviations by the opponent playing the strategy that is most harmful
to  the deviator.

We defer the formal details of the proof to Appendix \ref{proof-falk-thm}.
\qedhere 
\end{proof}

\subsection{Stackelberg Equilibria and CUE Outcomes}

In this subsection we study the relations between CUE outcomes and equilibria of a sequential (Stackleberg-leader) variant of the game.

We start 
by showing that every Pareto-efficient profile is a strong CUE
outcome, provided that no player can gain by becoming a Stackelberg
leader. Formally, 

\begin{prop}
\label{prop-efficient-Stackelberg-robust}Let $s$ be a strategy
profile that satisfies the following conditions:
\begin{enumerate}
  \item Pareto efficiency: if
$\pi_{i}\left(s'\right)>\pi_{i}\left(s\right)$, 
then $\pi_{-i}\left(s'\right)<\pi_{-i}\left(s\right)$ $\forall s'\in S$;
and
\item robustness to Stackelberg-leaders: if $s'_{-i}\in BR_{-i}\left(s_{i}'\right)$,
then $\pi_{i}\left(s'\right)\leq\pi_{i}\left(s\right)$ $\forall s'\in S$.
\end{enumerate}
Then $s$ is a strong CUE outcome.

\end{prop}
\begin{proof}{[sketch]}
Observe that profile $s$
is supported as a strong CUE by each player clustering together all
payoffs above her CUE payoff. This clustering implies that if a player
deviates, her opponent in any equilibrium of the new clustered game
either obtains at least her CUE payoff or she plays the best
reply to her unclustered utility. In the first (resp., second) case,  condition 1 (resp., 2)
implies that the deviator cannot gain.
See Appendix \ref{proof-sufficeint-strong} for details.
\end{proof}
\begin{example}
\label{exam-Stackelberg}We revisit the symmetric Cournot game of Example \ref{exm-weak-is-too-weak}: $S_{i}=\left[0,1\right]$
and $\pi_{i}\left(s_{i},s_{j}\right)=s_{i}\cdot\left(1-s_{i}-s_{-i}\right).$ Proposition \ref{prop-efficient-Stackelberg-robust} implies that the efficient profile $(0.25,0.25)$, in which the players equally share the monopoly quantity, is a strong CUE outcome. Observe that this profile, which induces each player a payoff of $\frac{1}{8},$ is robust to Stackelberg leaders, because a Stackelberg leader can obtain a payoff of at most $\frac{1}{8}$ (by the leader playing 0.5, and her opponent best replying by choosing 0.25).
\end{example}

We next show that any equilibrium of the sequential ``Stackelberg-leader''
variant of the underlying game is a CUE outcome.

We say that profile $s$ is a Stackelberg equilibrium if it is the
equilibrium outcome of a sequential variant of the game, in which
one of the players (the Stackelberg leader) plays first, and her opponent
observes the leader's strategy and best replies to it. Formally,
\begin{defn}
\label{def-Stackelberg}Strategy profile $s$ is a \emph{Stackelberg-equilibrium}
of the underlying game $G$ if there exists a player $i$ (the leader)
such that: 
\begin{enumerate}
\item the opponent best replies to the leader: $s_{-i}\in BR_{-i}\left(s_{i}\right)$,
and
\item the leader cannot achieve a higher payoff by deviating: $\pi_{i}\left(s'\right)\leq\pi_{i}\left(s\right)$
  for all profiles $s'$ for which $s'_{-i}\in BR_{-i}\left(s'_{i}\right)$.
\end{enumerate}
\end{defn}
Our next result shows that every Stackelberg equilibrium is a CUE outcome
that is supported by the leader (resp., follower) clustering together
all (no) payoffs. We defer the simple proof to Appendix
\ref{proof-stackelberg}.
\begin{prop}
\label{prop-Stackelberg}Let $s$ be a Stackelberg-equilibrium with player $i$ as the leader. Then $\left(\left(f_{i}^{\mathbb{R}},f_{-i}^{\emptyset}\right),s\right)$
is a CUE.
\end{prop}
%\ref{prop-Stackelberg}

\begin{example}
  \label{exam-Stackelberg1}We revisit the symmetric Cournot game of   Example \ref{exm-weak-is-too-weak} yet again.  Recall that
 $S_{i}=\left[0,1\right]$,
$\pi_{i}\left(s_{i},s_{j}\right)=s_{i}\cdot\left(1-s_{i}-s_{-i}\right)$.
It is well known that $\left(0.5,0.25\right)$ is the 
Stackelberg equilibrium of the game in which player 1 is the leader
and chooses her quantity first. Proposition \ref{prop-Stackelberg}
implies that $(0.5,0.25)$ 
is a CUE outcome. By clustering all of her payoffs together, player 1 commits to keep
playing 0.5 (as it is always a clustered best reply). Player 2, who
cannot gain anything from clustering, plays her unclustered
best reply 0.25.
\end{example}

\subsection{Constant-Sum and Common-Interest Games}\label{subsec:constant-sum}

In this section, we show a close connection between Nash equilibria
and CUE outcomes in both constant-sum games and common-interest games.

An underlying game is constant sum if the sum of payoffs is a constant. Formally: 
\begin{defn}
An underlying game $G=\left(\left\{ 1,2\right\} ,S,\pi\right)$ \emph{is
constant-sum} if $\pi_{1}\left(s\right)+\pi_{2}\left(s\right)=\pi_{1}\left(s'\right)+\pi_{2}\left(s'\right)\,\,\,\forall s.s'\in S.$

Recall that all Nash equilibria of a constant-sum game yield each
player her minimax (=maximin) payoff, that is, $s\in NE\left(G\right)$
implies that $\pi_{i}\left(s\right)=M_{i}\equiv \underline{M}_{i}=\overline{M}_{i}$. Observe that the constant
sum of payoffs must be equal to the sum of the minimax payoffs, that
is, $\pi_{1}\left(s\right)+\pi_{2}\left(s\right)=M_{1}+M_{2}$ for
any profile $s$.
\end{defn}
Next we show that (1) every Nash equilibrium of the underlying
constant-sum game is a strong CUE, and (2) every weak CUE outcome provides
each player an unclustered payoff that is equal to the game's value.
\begin{prop}\label{claim-constant-sum}
If the underlying game $G$ is  constant-sum, then
\begin{enumerate}
\item[(a)] every Nash equilibrium of $G$ is a strong CUE outcome;
\item[(b)] every weak CUE yields each player $i$ an unclustered payoff of $M_{i}$. 
\end{enumerate}
\end{prop}
The relatively straightforward proof of Proposition~\ref{claim-constant-sum} can be found in 
Appendix \ref{proof-constant-sum}. 

Next we show that the close connection between Nash equilibria and CUE outcomes
holds in the class of games in which all players have common interests, in the sense that their payoffs are always equal. 
\begin{defn}
A game has \emph{common interests} if $\pi_{i}\left(s\right)=\pi_{-i}\left(s\right)$
for each $s\in S$. 
\end{defn}

A strategy profile is \emph{Pareto-dominant} if it maximizes
the payoffs of all players, that is, $s$ is a Pareto-dominant profile
if $\pi_{i}\left(s\right)\geq\pi_{i}\left(s'\right)$ for each player
$i$ and profile $s'$. Observe that a common-interest game
admits at least one Pareto-dominant strategy profile, which must be
a Nash equilibrium. 

Our next result shows that the set of CUE outcomes coincides with
the set of Nash equilibria, and that the set of strong CUE outcomes
coincides with the set of Pareto-dominant Nash equilibria.
\begin{prop}\label{claim-common-interest}
If the underlying game $G$ has common interests, then
\begin{enumerate}
\item[(a)] $s$ is a CUE outcome iff it is a Nash equilibrium of $G$; 
\item[(b)] $s$ is a strong CUE outcome iff it is a Pareto-dominant
  Nash equilibrium 
of $G$.
\end{enumerate}
\end{prop}
The proof of Proposition~\ref{claim-common-interest} can be found in Appendix
\ref{proof-common-interest}). 

\section{CUE in Interval Games}\label{sec:interval-games}

In the previous section, we characterized CUE outcomes in which
both players play their unclustered best replies (i.e., Nash equilibria)
and CUE in which one of the players play her
unclustered best reply (i.e., Stackelberg-like equilibria). Arguably, the most interesting CUE outcomes (which
introduce new kinds of behavior) are the remaining set of CUE outcomes,
in which neither player plays her unclustered best reply. In this section,
we characterize this class of CUE outcomes under the 
widely applied assumption
that the game has monotone externalities. Specifically, we show that
in all these CUE, both players deviate from best replying in the
direction that is beneficial to the opponent.
One can interpret this result as showing that by  committing to a coarse utility the player commits to a favorable action vis a vis their co-player. 

\subsection{Games with Monotone Externalities}

An interval game has monotone externalities if increasing one's strategy
always affects the opponent's payoff in the same direction: either
positive externalities (increasing one's strategy increases the opponent's
payoff), or negative externalities (increasing one's strategy decreases
the opponent's payoff). Formally,
\begin{defn}
An interval game $G=\left(S,\pi\right)$ has \emph{monotone externalities}
if either:
\begin{enumerate}
\item \emph{Positive externalities}:
$\pi_{i}\left(s_i,s_{-i}\right)>\pi_{i}\left(s_i,s'_{-i}\right)$ for each player $i$, each strategy $s_i$, and each pair of strategies satisfying $s_{-i}>s'_{-i}$.
\item \emph{Negative externalities}: $\pi_{i}\left(s_i,s_{-i}\right)<\pi_{i}\left(s_i,s'_{-i}\right)$ for each player $i$, each strategy $s_i$, and each pair of strategies satisfying $s_{-i}>s'_{-i}$.
\end{enumerate}
\end{defn}
Games with monotone externalities are common in the economic literature (e.g., Cournot competition,
Bertrand competition with differentiated goods, Tullock competition,
public good games, etc.). 

We say that $s_{i}$ is \emph{externalities-higher} than $s'_{i}$ 
(or, equivalently, that  $s'_{i}$ is \emph{externalities-lower} than $s_{i}$),
and denote it by $s_{i}\succ_{-i}s'_{i}$, if the opponent gains by
player $i$ changing her strategy from $s'_{i}$ to $s_{i}$. Formally
\begin{defn}
Let $s_{i},s'_{i}\in S_{i}$ be two strategies in a game with monotone
externalities. We write $s_{i}\succ_{-i}s'_{i}$ (resp., $s_{i}\succeq_{-i}s'_{i}$)
if
\begin{enumerate}
  \item the game has positive externalities and $s_{i}>s'_{i}$ (resp.,
        $s_{i}\geq s'_{i}$); or
  \item the game has negative externalities and $s_{i}<s'_{i}$ (resp.,
    $s_{i}\leq s'_{i}$). 
\end{enumerate}
\end{defn}

Theorem \ref{Thm:neccesary-conditions-monotone} characterizes the CUE outcomes in which neither
player plays her unclustered best reply. It shows that in all such CUE outcomes:
\begin{enumerate}
\item 
each player's CUE behavior treats her opponent better  than unclustered best reply behavior would do;  and
\item 
any externalities-lower profile cannot Pareto dominates the CUE outcome, and if it improves the unclustered payoff of one of the players then the opponent cannot be playing an unclustered best reply. (We interpret this latter condition as requiring that no player can
gain by becoming a Stackelberg leader and deviating to an
externalities-lower strategy.) 
\end{enumerate}

\begin{thm}
\label{Thm:neccesary-conditions-monotone}Let $G$
be an interval game with monotone externalities. Let $s$ be an interior
CUE outcome in which $s_{i}\notin BR_{i}\left(s_{-i}\right)$ for
each player $i$. Then:
\begin{enumerate}
\item $s_{i}\succ_{-i}BR_{i}\left(s_{-i}\right)$
for each player $i$;
\item if $s'_{i}\preceq_{-i}s_{i}$, $s'_{-i}\preceq_{i}s_{-i}$, and either (a) $\pi_{-i}\left(s'\right)\geq\pi_{-i}\left(s\right)$ or (b) $s'_{-i}\in BR_{-i}\left(s_{i}'\right)$, then $\pi_{i}\left(s'\right)\leq\pi_{i}\left(s\right)$.
\end{enumerate}
\end{thm}
\begin{proof}{[Sketch]}
For part (1), assume to the contrary that 
$s_{i}\prec_{-i}BR_{i}\left(s_{-i}\right)$. Consider a  sufficiently small deviation of
player $-i$ toward her unclustered best reply (which can be implemented by slightly altering her
clustering). Any payoff-improving reaction of player $i$ must be
toward player $i$'s best reply, which is the direction that is
beneficial to player $-i$ due to monotone externalities. Thus, player $-i$ gains from  the deviation
and $s$ cannot be a CUE outcome. This proves part (1). In order to
prove part (2),  
assume to the contrary that there exists a strategy profile
$s'$ that violates condition (2).  One can show that player $i$ can
change her clustering and improve her clustered payoff by changing
her strategy to $s'_{i}$; this is followed by at most one
additional stage in the improvement path resulting in the plausible
equilibrium $s'$, which violates $s$ being a CUE outcome.\\ 
The details of the proof are in Appendix 
\ref{proof-montotone-externalities}. 
\end{proof}
\begin{rem} \label{rem-theorem-1} Theorem
  \ref{Thm:neccesary-conditions-monotone} shows that CUE yields
    results that are qualitatively different  from most existing related
    solution concepts (e.g., clustered preferences
  (\citealp{heifetz2007dynamic}), delegation
  (\citealp{fershtman1987equilibrium}), biased beliefs
  (\citealp{heller2020biased}), and naive analytics
  (\citealp{berman2021naive})). All these existing solution
    concepts predict that players treat their opponents worse than
    they would using their unclustered best replies in games with strategic
    substitutes. By contrast, we have the opposite prediction in all
    CUEs in which neither player plays her unclustered best reply. The
    key difference between our result and theirs is induced by two
    novel aspects of our solution concept:  
  \begin{enumerate}
    \item  CUE allows the subjective (clustered) payoffs to differ from the
      material (unclustered) ones only by clustering 
      payoffs in an interval. This implies
      that the direction that improves one's subjective payoffs is the
      same direction that improves her material payoffs, which  is the
      driving force behind Theorem
      \ref{Thm:neccesary-conditions-monotone}. By contrast, the
        existing solution concepts allow the subjective preferences
        to substantially differ from the material ones, which allows
        an agent's deviation to increase her subjective payoff while
        decreasing her material payoff. 
    \item Most other solution concepts imply that a player with
      strictly concave material payoffs have a unique subjective best
      reply, which is a key argument in ruling out equilibrium
      behavior in which players treat their opponents better than
         they would using their material best replies in games with strategic
            substitutes. By contrast, CUE induces  players to be
      indifferent between an interval of subjective best-reply
      strategies, which allows the players to treat their       opponents better than they would using their unclustered best
            replies.  
\end{enumerate}

\end{rem}

\subsection{Games with Strategic Complements}
In this subsection, we study  games with strategic complements,  
and show that for these games, the three conditions
of Theorem \ref{Thm:neccesary-conditions-monotone} fully characterize
CUE outcomes; that is, they provide both necessary and sufficient
conditions for CUE outcomes.

A game $G=\left(S,\pi\right)$ \emph{has strategic
complements} if $\frac{\partial\pi_{i}\left(s\right)}{\partial s_{i}}$
is strictly increasing in $s_{j}$ for each player $i$ and each strategy
$s_{i}$.  Games with strategic complements are common in the economic
literature, and include, in particular, price competition with differentiated
goods (Example \ref{exam-Price-competition}).
It is well known that every game $G$ with strategic complements and
monotone externalities admits a \emph{worst Nash equilibrium}
$s^{WNE}\in NE\left(G\right)$,  
in which all players play their externalities-lowest equilibrium startegies, i.e., $s_{i}^{WNE}\preceq_{-i}s_{i}^{NE}$ for every Nash
equilibrium $s^{NE}\in NE\left(G\right)$ and player $i\in I$
(see, e.g., \citealp{milgrom1990rationalizability}).
It is well-known that the worst Nash equilibrium is Pareto dominated by all other Nash equilibria of the game.

Theorem \ref{thm-complements} characterizes the set of CUE outcomes
in monotone games with strategic complements. 
It shows that the two necessary conditions for being a CUE outcome in an internal game with monotone externalities in Theorem \ref{Thm:neccesary-conditions-monotone} are also sufficient conditions if the game has strategic complements. 
This characterization implies, in particular, that all CUE outcomes
have externalities-higher strategies and higher payoffs relative to
the worst Nash equilibrium. 

\begin{thm}
\label{thm-complements}Let $G$
be an interval game with monotone externalities and strategic complements. Let $s$ be an interior
strategy profile. Then $s$ is a CUE outcome iff:
\begin{enumerate}
\item $s_{i}\succeq_{-i}BR_{i}\left(s_{-i}\right)$
for each player $i$;
\item if $s'_{i}\preceq_{-i}s_{i}$, $s'_{-i}\preceq_{i}s_{-i}$, and either (a) $\pi_{-i}\left(s'\right)\geq\pi_{-i}\left(s\right)$ or (b) $s'_{-i}\in BR_{-i}\left(s_{i}'\right)$, then $\pi_{i}\left(s'\right)\leq\pi_{i}\left(s\right)$.
\end{enumerate}
Moreover, profile $s$ has externalities-higher strategies
and higher payoffs than the worst Nash equilibrium
(i.e., $s_{i}\succeq_{-i}s_{i}^{WNE}$ and $\pi_{i}\left(s\right)\succeq_{-i}\pi_{i}\left(s^{WNE}\right)$
 $\forall i$).

\end{thm}

\begin{proof}{[Sketch]}
For the ``if'' direction,
we show that $s$ can be supported as a CUE
outcome if each player clusters all payoffs above her payoff in
$s$ (i.e., players satisfice with an aspiration level
equal to the equilibrium payoff). Because the game has strategic
complements, condition (1) 
implies that all the stages in an improvement path must be in the
externalities-lower direction. Given the players' clustering, the
improvement path must end in either (a) a Pareto-dominant profile, or (b) a profile in which the non-deviating player plays her
unclustered best reply. Condition (2) implies that the deviator cannot
gain in either of these cases.
For the ``only if'' direction, it is relatively simple to show that
strategic complements allow us to extend the argument of Theorem
\ref{Thm:neccesary-conditions-monotone} to cases in which one of the
players plays her unclustered best reply. 

For the claim in the final sentence of the theorem statement, note
that 
the inequality $s_{i}\succeq_{-i}BR_{i}(s_{-i})$
implies that $s_{i}\succeq_{-i}s_{i}^{WNE}$ for each player $i$
by a standard property of games with strategic complements (proved
in Lemma \ref{Lemma-complements-standard}).  It remains 
to show that $\pi_{i}\left(s\right)\geq\pi_{i}\left(s^{WNE}\right)$
for each player $i$. Assume to the contrary that  $\pi_{1}\left(s\right)<\pi_{1}\left(s^{WNE}\right).$
We consider that  a deviation by player 1 to $f_1^\emptyset$ (not clustering any
payoffs together), followed by an improvement path in which
the players sequentially decrease their strategies into best replying
until they converging to a plausible equilibrium. One can show that
strategic complements imply that player 1 obtains a payoff of at least
$\pi _1(s^{WNE})> \pi_1(s)$ in this plausible equilibrium, which
contradicts  the fact that $s$ is a CUE outcome. \\
The details of the proof are in Appendix 
\ref{proof-complements}. 
\end{proof}

Next, we apply Theorem \ref{thm-complements} to
price competition with differentiated goods (the linear city model
$\grave{\textrm{a}}$ la Hotelling). Specifically, we show that in all CUE outcomes both players  choose prices and obtain payoffs at least as high as in the unique Nash equilibrium.

\begin{example}[Price competition
with differentiated goods; adapted from see the textbook
analysis in \citealp
{mas1995microeconomic}, Section 12.C.]\label{exam-Price-competition}
Consider a mass one of consumers uniformly distributed in the interval $\left[0,1\right]$.
Consider two firms that produce widgets, located at the two extreme
locations: 0 and 1. Every consumer wants at most one widget. Producing
a widget has a constant marginal cost, which we normalize to be zero.
Each firm $i$ chooses price $s_{i}\in\left[0,M\right]$ for its
widgets. The total cost of buying a widget from firm
$i$ is equal to its price $s_{i}$ plus $t$ times the consumer's
distance from the firm, where $t\in\left(0,M\right)$. Each buyer buys a widget from the firm with
the lower total buying cost. This implies that the total demand for
widget $i$ is given by function $q_{i}\left(s_{i},s_{-i}\right)$,
where 
\[
q_{i}\left(s_{i},s_{-i}\right)=\begin{cases}
\begin{array}{cc}
0 & \mbox{ if }\frac{s_{-i}-s_{i}+t}{2\cdot t}<0,\\
\frac{s_{-i}-s_{i}+t}{2\cdot t} & \mbox{ if }0<\frac{s_{-i}-s_{i}+t}{2\cdot t}<1,\\
1 & \mbox{ if }\frac{s_{-i}-s_{i}+t}{2\cdot t}>1.
\end{array}\end{cases}
\]
The payoff (profit) of firm $i$ is given by
$\pi_{i}\left(s\right)=s_{i}\cdot q_{i}\left(s\right)$. 
Observe that no strategy profile $s$ is Pareto dominated by a lower
profile $s'<s$. This is because $s'_{i}\cdot
q_{i}\left(s'\right)=\pi_{i}\left(s'\right)\geq\pi_{i}\left(s\right)=s_{i}\cdot
q_{i}\left(s\right)\Rightarrow
q_{i}\left(s'\right)>q_{i}\left(s\right)$. 
Thus, if $s'$ Pareto dominates $s$, then 
$q_{i}\left(s'\right)>q_{i}\left(s\right)$ 
and $q_{-i}\left(s'\right)>q_{-i}\left(s\right)$, a
contradiction. 

It is well-known that the game has strategic complements, and that
each player has a unique best reply for each opponent's strategy,
which is given by $BR_{i}\left(s_{-i}\right)=\max(\frac{s_{-i}+t}{2},s_{-i}-t)$. 
This implies that  both players play weakly above
their unclustered best-replies iff $s_{1}\in\left[\frac{s_{2}+t}{2},2s_{2}-t\right]$
(which implies, in particular, that $s_{1},s_{2}\geq t$).

Observe that the payoff of player $i$ when playing strategy $s_{i}\leq3t$
and facing a best-replying opponent is given by $\pi_{i}\left(s_{i},BR_{i}\left(s_{i}\right)\right)=\frac{s_{i}\left(1.5t-0.5s_{i}\right)}{2\cdot t}.$
Thus, $\pi_{i}\left(s_{i},BR_{i}\left(s_{i}\right)\right)$ is a strictly
concave function of $s_{i}$ with a unique maximum at $s_{i}=1.5t$.
This implies that profile $s$ is robust to (externalities-)lower
Stackelberg-leaders iff, for each player $i$, either
$s_{i}\geq1.5t\,\,\textrm{and}\,\,\pi_{i}\left(s\right)\geq\frac{9}{16}t$, 
or $s_{i}\leq\min\left(1.5t,2s_{-i}-t\right)$.  By combining these
inequalities with Theorem \ref{thm-complements}, we get that a strategy
profile $s$ is a CUE outcome iff
\begin{enumerate}
\item Each strategy is higher than the unclustered best reply: $s_{1}\in\left[\frac{s_{2}+t}{2},2s_{2}-t\right]$;
and
\item if $s_{i}\geq1.5t$, then player $i$'s unclustered payoff is further
required to be above $\frac{9}{16}t$ (her payoff if she were a Stackelberg
leader).

\end{enumerate}

Figure \ref{fig:The-Set-of-differiantaed-goods} demonstrates the
set of CUE outcomes for $t=1$ and $M=3$. In all the CUE both players set higher prices and obtain higher payoffs than in the Nash equilibrium.

\begin{figure}[ht]
\centering{}\caption{\label{fig:The-Set-of-differiantaed-goods}The Set of CUE Outcomes
in Example \ref{exam-Price-competition} ($t=1$, $M=3$)}
\includegraphics[scale=0.52]{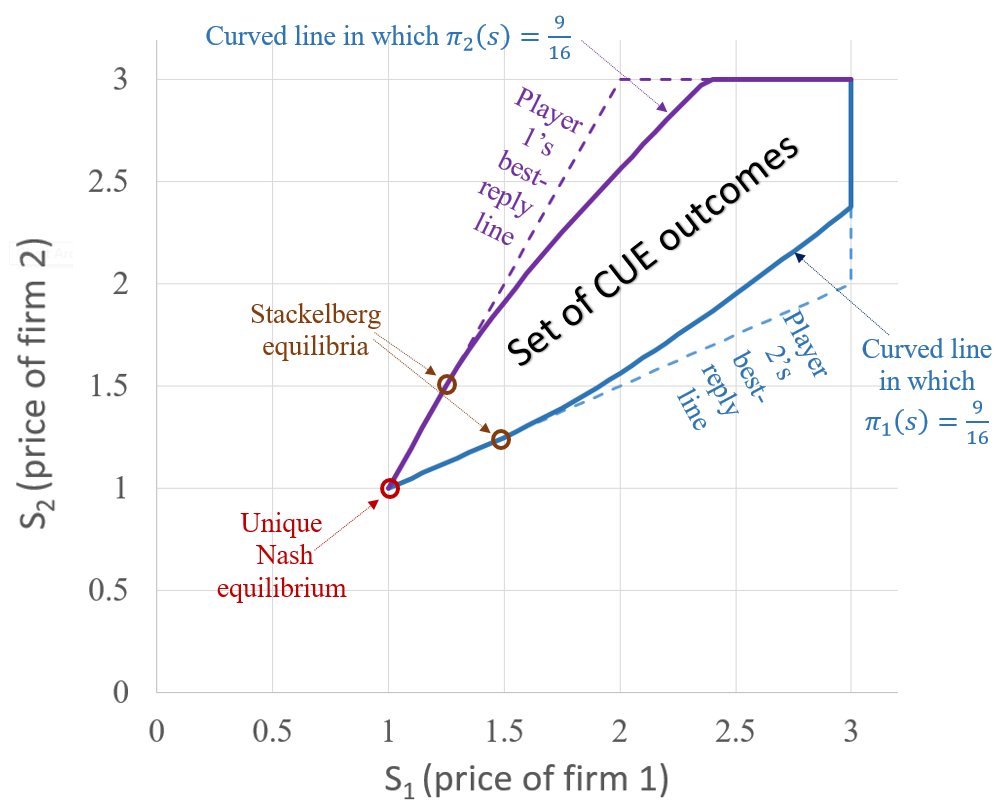}
\end{figure}
\end{example}

\subsection{Games with Strategic Substitutes\label{sec:Wishful-thinking-and-substitute}}
The set of CUE outcomes can be divided to two disjoint classes: (1)
CUE outcomes in which at-least one of the players plays her unclustered
best reply, and (2) CUE outcomes in which neither player plays her
unclustered best reply. Theorem \ref{thm-complements} shows that in games
with strategic complements, both classes induce similar behavior,
which deviates from Nash behavior in the direction that is beneficial
to the players. We now show that the two classes induce
qualitatively different behaviors in games with strategic
substitutes. Theorem \ref{Thm:neccesary-conditions-monotone} implies
that in the second class of CUE outcomes both players deviate from
unclustered best reply in the direction that is beneficial to the
opponent. 

By contrast, Theorem \ref{thm-substitues} shows that in the 
first class of CUE outcomes (in which at least one of the players
plays her unclustered best reply), the non-best-replying player $i$
deviates from her unclustered best reply in the direction that is harmful
to the opponent. 
Hence the player that does not best reply uses the clustering as a
threat rather than a commitment for a favorable action.  Moreover,
under the additional mild assumption of the 
payoff function being strictly (rather than only weakly) concave,
there exists a Nash equilibrium in which  player $i$'s CUE strategy is
externalities-lower than her Nash equilibrium strategy (while the
opposite holds for player $-i$). 
\begin{thm}\label{thm-substitues}
Let $G$ be an interval game with monotone externalities and strategic substitutes. Let $s$ be an interior CUE outcome. Assume that $s_{-i}\in BR_{-i}(s_i)$. Then:
\begin{enumerate}
    \item 
    $s_i \preceq_{-i} BR_i(s_{-i}) $
    \item If the payoff function is strictly concave in the player's own strategy, then there exists a Nash equilibrium  $s^{NE}\in NE(G)$, such that $s_i \preceq s^{NE}_i$ and   $s_{-i} \succeq s^{NE}_{-i}.$
\end{enumerate}
\end{thm}

\begin{proof} {[Sketch]}
\begin{enumerate}
    \item Assume to the contrary that $s_i\succ BR_i(s_{-i})$. Let
      $s'_i \prec s_i$ be a nearby externality-lower strategy. Player
      $i$ can increase her unclustered payoff by deviating to
      clustering all payoffs above $\pi_i(s'_i,s_{-i})$. This
      deviation induces an improvement path in which player $i$
      increases her unclustered payoff by  changing her
      strategy to $s'_i$. Since the game has strategic substitutes, an
      opponent's reaction must be in the externality-higher direction,
      which further increases player $i$'s payoff. 
    
    \item Consider an auxiliary game in which player $i$ is restricted
      to choosing strategies that are weakly externalities-lower than
      $s^i$. It is straightforward to show that the restricted game
      admits a Nash equilibrium $s^{NE}$, and that
      since the game has strategic complements and the unclustered
      utilities are strictly concave, the profile $s^{NE}$
      has to be a Nash equilibrium of the original game, and that it
      must satisfy $s_i \succeq BR_{-i}(s^NE_{-i})$ and $s_i\preceq_i
      s^{RE}_i.$
\end{enumerate}
See Appendix~\ref{proof-substitutes} for details. 
\end{proof}
Thus, the qualitative predictions are different in the two classes of
CUE outcomes. In the first class, both players deviate from unclustered
best replying in the direction that is beneficial to the opponent. in
the second class, only one of the players deviate from unclustered best
reply, and it does so in the direction that is harmful to the
opponent. Taken together, Theorems
\ref{Thm:neccesary-conditions-monotone}--\ref{thm-substitues} imply
that CUE predicts cooperative outcomes in which players treat each
other better than the unclustered best replies in games with strategic
complements, while the prediction for games with strategic 
%Y(24.4): corrected "complements" to "substitues"
substitutes is ambiguous, and depends on whether one or both players cluster
payoffs together. Experimental evidence  supporting our theoretical
predictions is presented 
by \citet{potters2009cooperation}, 
{who show that there is significantly
more
cooperation in two-player interval games with strategic complements
than in games with strategic substitutes, and  
by} \citet{suetens2007bertrand}, {who 
present similar results for oligopoly experiments.}
Our results are demonstrated in the following example of Cournot competition. 

\begin{example}
  \label{exam-Cournot} 
   We expand on the discussion of the
  symmetric Cournot game presented in Examples
    \ref{exm-weak-is-too-weak}--\ref{exam-Stackelberg1}.  Recall that
  $S_{i}=\left[0,1\right]$,  
$S_{i}=\left[0,1\right]$,
    $\pi_{i}\left(s_{i},s_{j}\right)=s_{i}\cdot\left(1-s_{i}-s_{-i}\right)$,
    the best-reply function of each player $i$ is given by
    $BR_i(s_{-i})=\frac{1-s_{-i}}{2},$ and that the payoff of player $i$ when
choosing quantity $s_i$ and facing a best-replying opponent is
$\frac{s_i\cdot (1-s_i)}{2},$ which has a unique maximal payoff of 0.5
obtained by choosing the Stackelberg-leader quantity $s_i=0.5$.  

We begin by characterizing the CUE in which one of the players (player
$-i$) plays her unclustered best reply. Theorem \ref{thm-substitues}
implies that $s_i\geq \frac{1}{3}.$ Observe that $s_i$ cannot be
larger than 0.5, because otherwise player $i$ would gain by deviating
to  clustering payoffs above $\pi_i(0.5,s_{-i})$ and following the
improvement path that starts by changing her strategy to 0.5. Any $s_i
\in [\frac{1}{3},\frac{1}{2}]$ can be supported as such a CUE outcome
by having player $i$ clustering all payoffs and player $-i$
clustering the payoffs below her CUE payoff. 

Next, we characterize the CUE in which neither player plays her
unclustered best reply. Theorem \ref{Thm:neccesary-conditions-monotone}
implies that each player chooses a lower quantity than her unclustered
best reply.  If $s_i<\min(s_{-i},BR_i(s_{-i}))$, then
$(s_i,s_{-i})$ cannot be a CUE outcome, because for a sufficiently
small $\epsilon>0,$ player $i$ gains by clustering the payoffs above
$\pi_i(1-s_i-s_{-i}-\epsilon,s_{-i}),$ and changing her strategy to
$1-s_i-s_{-i}-\epsilon.$  Any payoff-improving opponent's reaction
reaction must be to a lower quantity, which further benefits player
$i$. Next observe that any symmetric profile $(s_i,s_{i})$ cannot be a
CUE outcome for either (1) $s_i>\frac{1}{3}$, because the players play
above their unclustered best reply, and (2) for
any $s_i<\frac{1}{4}$, because player $i$ gains by deviating to
clustering the payoffs above $\pi_i(0.5,s_i)$ and  deviating to
0.5. Finally, note that a symmetric profile $(s_i,s_i)$ can be
supported as a CUE outcome for all $s_i\in[\frac{1}{4},\frac{1}{3}]$
by having each player cluster together the payoffs above the CUE payoff
$\pi_i(s_i,s_i).$  

Thus the set of CUE outcomes (illustrated in Figure
\ref{fig:The-Set-of-CUE-Cournot}) includes 3 intervals that intersect
in the unique Nash equilibrium $\frac{1}{3}$: two intervals that end
in the Stackelberg equilibria, in which one of the players plays her
unclustered best reply and the sum of payoffs is lower than in the Nash
equilibrium; and an interval that ends in the efficient profile
$\frac{1}{4}$ in which both players equally divide the monopoly
quantity. In the latter interval, both players gain a higher payoff
than in the Nash equilibrium. 
\begin{figure}[ht]
    \centering{}\caption{\label{fig:The-Set-of-CUE-Cournot}The Set of
    CUE Outcomes 
in Example \ref{exam-Cournot} (Cournot competition)}
\includegraphics[scale=0.48]{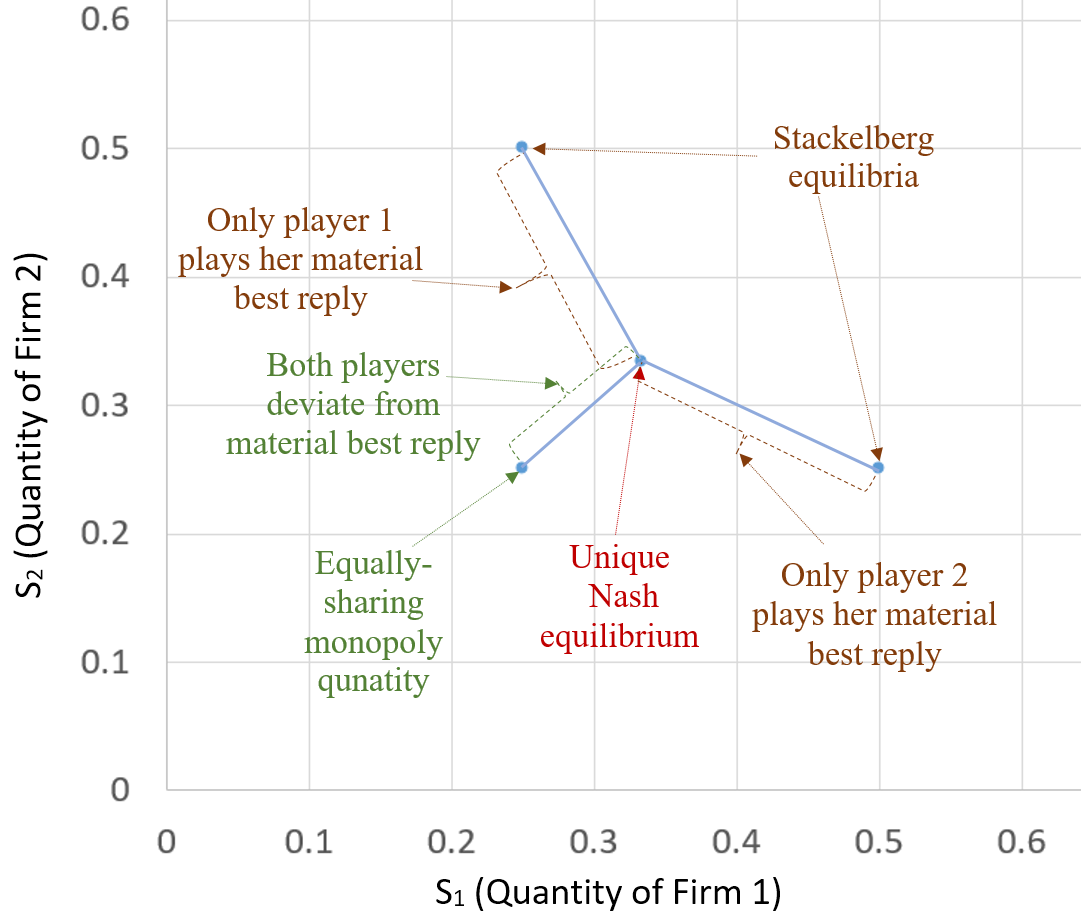}
\end{figure}
\end{example}

\shortv{
  \bibliographystyle{ACM-Reference-Format}
  \bibliography{CUE}
  }

{
\section{Conclusion}
We have considered the strategic implications of coarsening utilities by
clustering payoffs together.  This led us to a new solution concept,
coarse-utility equilibrium (CUE).  Clustering captures a common
human phenomenon: the fact that people describe outcomes using terms
like ``good''/''bad'', or ``unacceptable''/''fair''/''generous''.
Perhaps not surprisingly, CUE is able to capture in a reasonably
natural way cooperation in prisoner's dilemma, by assuming that people
use satisficing and cluster all outcomes above their aspiration level.
Less trivially, we show that CUE makes interesting predictions that
are supported by experimental evidence in interval games with montone
externalities.}

{
Of course, many explanations besides clustering can explain players'
behaviors in the games that we consider.  We need experimental
evidence to verify that clustering is really what is going on.
Fortunately, our model should be readily testable.
For example, by
  designing experiments involving some of the games we discuss in the
  paper, a treatment in which the play of the
  game is preceded by a stage of pre-play communication can be
  compared to a
  treatment in which the game is played without communication. Such
  comparisons can reveal the extent to which the difference between
  Nash equilibria and the equilibria predicted by our model is
  confirmed in the lab. The pre-play communication can be designed
  to permit only messages regarding clustering, so as to examine the extent
  to which players' choice of clustering is consistent with our
  model.  We hope to carry out these experiments in future work.}

\appendix
%dummy comment inserted by tex2lyx to ensure that this paragraph is not empty
%\newpage
%\setcounter{page}{1}
\section*{Appendix}
\section{Proofs}
\subsection{Proof of Proposition \ref{prop-falk-thm} (Folk Theorem for
Weak CUE)}\label{proof-falk-thm}
The following simple observation will be useful in the proof of Proposition \ref{prop-falk-thm}. 
(The standard proof is omitted for brevity.)
\begin{lem}
\label{fact-minimax}~
\begin{enumerate}
\item The 
%Y: minimax corrected to maximin due to reviewer 1's comment
maximin payoff $\underline M_{i}$ depends only on the payoff function of
  player $i$, and not on the payoff function of the opponent.

\item Each player must obtain at least her 
maximin payoff in
    all Nash   equilibria; 
that is, if $s\in NE\left(G\right)$ then $\pi_{i}\left(s\right)\ge \underline M_{i}$.
\end{enumerate}
\end{lem}

\begin{proof}[Proof of Proposition \ref{prop-falk-thm}]
~
\begin{enumerate}
\item Assume to the contrary that $s$ is a weak CUE outcome and that it
is not individually rational. Let $\left(f,s\right)$ be a weak CUE.
Let $i$ be a player for which $\pi_{i}\left(s\right)<\underline M_{i}$. Consider
a deviation by player $i$ to $f'_{i}=Id_{i}$ (i.e., not clustering
any payoffs together). Let $s'$ be a Nash equilibrium of
$G_{\left(Id_{i},f_{-i}\right)}$.
Lemma \ref{fact-minimax} implies that $\pi_{i}\left(s'\right)\geq
\underline M_{i}$, 
which contradicts the assumption that $\left(f,s\right)$ is weak
CUE outcome.
\item Assume that $s$ is strongly individually rational. Let $f^{\mathbb{R}}$
be the symmetric profile in which all players cluster together all
payoffs. It is immediate that $s\in NE\left(G_{f^{\mathbb{R}}}\right)$.
Fix an arbitrary player $i$. Let $\overline{s}_{-i}\in S_{-i}$ be
the strategy profile that guarantees that player $i$'s payoff is
at most $\overline{M}_{i}$ (i.e., $\pi_{i}\left(s''_{i},\overline{s}_{-i}\right)\leq \overline M_{i}$
for each $s''_{i}\in S_{i}$). For each clustering $f'_{i}$, let
$s'_{i}$ be a clustered best reply of player $i$ against $\overline{s}_{-i}$.
Observe that $\left(s'_{i},\overline{s}_{-i}\right)\in NE\left(G_{\left(f'_{i},f_{-i}\right)}\right)$
and that $\pi_{i}\left(s'_{i},\overline{s}_{-i}\right)\leq \overline M_{i}\leq\pi_{i}\left(s\right)$,
which implies that $\left(f^{\mathbb{R}},s\right)$ is a weak CUE.\qedhere
\end{enumerate}
\end{proof}

\subsection {CUE in Constant-Sum Games}\label{proof-constant-sum}
\begin{proof}[Proof of Proposition \ref{claim-constant-sum}]
~
\begin{enumerate}
\item Let $s\in NE\left(G\right).$ We show that $\left(f^\emptyset,s\right)$ is
a strong CUE. Fix
a player $i$ and any clustering $f_{i}$. By Lemma \ref{fact-minimax} in
Appendix \ref{proof-falk-thm}: 
\[
s'\in NE\left(G_{\left(f_{i},Id_{-i}\right)}\right)\Rightarrow\pi_{-i}\left(s'\right)\geq M_{-i}\Rightarrow\pi_{i}\left(s'\right)\leq M_{i}\leq\pi_{i}\left(s\right)\Rightarrow\left(f,s\right)\,\textrm{is a strong CUE.}
\]
\item Let $\left(f,s\right)$ be a weak CUE. Proposition 1 implies that
$\pi_{i}\left(s\right)\geq M_{i}$ for each $i\in I$. The game being
  constant-sum implies
that $\pi_{1}\left(s\right)+\pi_{2}\left(s\right)=M_{1}+M_{2}$, which
implies that $\pi_{i}\left(s\right)=M_{i}$ for each $i\in I$. \qedhere
\end{enumerate}
\end{proof}
The following example shows that an underlying zero-sum game might
admit a strong CUE outcome that is not a Nash equilibrium (although
it will provide each player her unique Nash equilibrium payoff). 
\begin{example}[Non-Nash strong CUE outcome in a zero-sum Game]
Consider the zero-sum game $G_{zs}$ 
that is presented in Table \ref{tab:Underlying-Zero-Sum-Game}
below. Consider the symmetric clustering profile $f^{\geq0}$ in which
the players cluster all non-negative payoffs together. We show that
$\left(f^{\geq0},\left(b,b\right)\right)$ is a strong CUE, although
$\left(b,b\right)$ is not a Nash equilibrium of $G_{zs}$.
Observe first that $(b,b)\in NE(G_{f^{\geq0}})$. Next, consider a
deviation by player $i$ to a clustering $f'_{i}$. Let $s'\in
NE(G_{\left(f'_{i},f_{-i}^{\geq0}\right)})$ 
be a Nash equilibrium of the coarse-utility game $G_{\left(f'_{i},f_{-i}^{\geq0}\right)}$.
Observe that the opponent can guarantee a clustered payoff of at
least 0 in $G_{\left(f'_{i},f_{-i}^{\geq0}\right)}$ by playing $a$.
This implies that $f_{-i}\left(\pi_{-i}\left(s'\right)\right)\geq0\Rightarrow\pi_{-i}\left(s'\right)\geq0$.
The fact that the game is zero sum implies that $\pi_{i}\left(s'\right)\leq0$,
and, thus $\left(f^{\geq0},\left(b,b\right)\right)$ is a strong CUE.
\begin{table}[ht]
\caption{\label{tab:Underlying-Zero-Sum-Game}Underlying Zero-Sum Game $G_{zs}$}

\centering{}
\begin{tabular}{|c|c|c|c|}
\hline 
 & \textcolor{red}{a } & \textcolor{red}{b } & \textcolor{red}{c}\tabularnewline
\hline 
\textcolor{blue}{a } & \textcolor{blue}{0},\textcolor{red}{0}  & \textcolor{blue}{0},\textcolor{red}{0}  & \textcolor{blue}{0},\textcolor{red}{0}\tabularnewline
\hline 
\textcolor{blue}{b } & \textcolor{blue}{0},\textcolor{red}{0} & \textcolor{blue}{0},\textcolor{red}{0}  & \textcolor{blue}{-1},\textcolor{red}{1}\tabularnewline
\hline 
\textcolor{blue}{c } & \textcolor{blue}{0},\textcolor{red}{0}  & \textcolor{blue}{1},\textcolor{red}{-1}  & \textcolor{blue}{0},\textcolor{red}{0}\tabularnewline
\hline 
\end{tabular}
\end{table}
\end{example}

\subsection{Proof of Proposition \ref{claim-common-interest} (Games with
Common Interests)}\label{proof-common-interest} 

\begin{enumerate}
\item Proposition \ref{fact2:any-NE-is-NE-clustered-1} implies that $s\in
  NE\left(G\right)$ so $s$ 
is a CUE outcome. Next, assume to the contrary that $\left(f,s\right)$
is a CUE and that $s\neq NE\left(G\right)$. The fact that $s\neq NE\left(G\right)$
implies that there exists player $i\in I$ and strategy $s'_{i}$,
such that $\pi_{i}\left(s\right)<\pi_{i}\left(s'_{i},s_{-i}\right)$.
Consider a deviation by player $i$ to $Id_{i}$ (i.e., to not clustering
any payoffs). Observe that $s\notin NE\left(G_{\left(Id_{i},f_{-i}\right)}\right),$
and that the fact that the game has common interests implies that
the payoffs of all players strictly improve in an improvement path.
Consider the improvement path in which at each stage one of the players
who is not best-replying changes her strategy to her clustered best
reply. The improvement path cannot have a cycle (since the payoffs
of all players strictly increase) and it must converge 
to some $s''\in
PNE\left(G_{\left(Id_{i},f_{-i}\right)},s\right)$. 
It is immediate that $\pi_{i}\left(s''\right)>\pi_{i}\left(s\right)$,
which contradicts $\left(f,s\right)$ being a CUE.
\item If $s$ is a Pareto-dominant Nash equilibrium, then it is immediate
that $\left(Id,s\right)$ is a strong CUE because $\pi_{i}\left(s'\right)\leq\pi_{i}\left(s\right)$
for any strategy profile $s'$. Next, assume to the contrary that
that $\left(f,s\right)$ is a strong CUE and $s$ is not a Pareto-dominant
Nash equilibrium of $G$. Let $s'$ be a Pareto-dominant Nash equilibrium

of $G$. Proposition \ref{fact2:any-NE-is-NE-clustered-1} implies that $s'\in
NE\left(G_{f'}\right)$ 
for any clustering profile $f'$. This implies that if a player $i$
deviates to a clustering $f''_i$, then $s'\in
NE\left(G_{\left(f_{i}'',f_{-i}\right)}\right)$ 
and $\pi_{i}\left(s'\right)>\pi_{i}\left(s\right)$, which contradicts
$\left(f,s\right)$ being a strong CUE. 
\end{enumerate}

\subsection {Proof of Proposition \ref{prop-efficient-Stackelberg-robust} ( Conditions Implying Strong CUE)}\label{proof-sufficeint-strong}

Let $f^{\geq\pi\left(s\right)}=\left(f_{i}^{\geq\pi_{i}\left(s\right)},f_{-i}^{\geq\pi_{-i}\left(s\right)}\right)$
be the profile in which each player clusters all the payoffs above
her CUE payoff. We show that $\left(f^{\geq\pi\left(s\right)},s\right)$
is a strong CUE. Assume to the contrary that
$\left(f^{\geq\pi_{i}\left(s\right)},s\right)$ 
is not a strong CUE. Then there exists a player $i$, a clustering
$f'_{i}$, and an equilibrium $s'\in NE(G_{(f'_{i},f_{-i}^{\geq\pi_{-i}\left(s\right)})})$
such that $\pi_{i}\left(s'\right)>\pi_{i}\left(s\right)$. The fact
that $s'\in NE(G_{(f'_{i},f_{-i}^{\geq\pi_{-i}\left(s\right)})})$
implies that either $\pi_{-i}\left(s'\right)\geq\pi_{-i}\left(s\right)$
or $s'_{-i}\in BR_{-i}\left(s_{i}'\right)$, which contradicts either
condition (1) or condition (2) above.

\subsection{Proof of Proposition \ref{prop-Stackelberg} (Stackelberg Equilbirum is a CUE)\label{proof-stackelberg}}
  The fact that player $i$ clusters together all payoffs implies that
  he would continue playing $s_{i}$ in any plausible equilibrium following
  a deviation by her opponent to a different clustering. Due to this,
  condition $1$ of Definition \ref{def-Stackelberg} implies that player
$-i$ cannot gain by deviating to a different clustering. The fact
that player $-i$ does not cluster any payoffs together implies that
she always best replies against her opponent, which, due to condition
2 of Definition \ref{def-Stackelberg}, implies that player $i$ cannot
gain by deviating to a different clustering. This implies that $\left(\left(f_{i}^{\mathbb{R}},f_{-i}^{\emptyset}\right),s\right)$
is a CUE.

\subsection {Proof of Theorem \ref{Thm:neccesary-conditions-monotone} (Games with Monotone Externalities)}\label{proof-montotone-externalities}

Let $f$ be a clustering
profile $f$ for which $\left(f,s\right)$ is a CUE.\\
Part (1): Assume to the contrary that $s_{i}\not \succ_{-i}BR_{i}\left(s_{-i}\right)$.
The assumption that $s_{i}\notin BR_{i}\left(s_{-i}\right)$ implies that $s_{i}\prec_{-i}\left(BR_{i}\left(s_{-i}\right)\right)$.
Let $s'_{-i}\neq s_{-i}$ be a strategy that satisfies the following
two properties: (1) $s'_{-i}$ is closer to $BR_{-i}\left(s_{i}\right)$
than $s_{-i}$, and (2) $s'_{-i}$ is sufficiently close to $s_{-i}$
that $s_{i}\prec_{-i}BR_{i}\left(s'_{-i}\right)$ 
%Y(17.10): new sentence in brackets due to R1
(such a strategy $s'_{-i}$ exists because the set of strategies that are (strictly) externalities-lower than
$BR_{i}(s'_{-i})$ is open). Let
$\pi'_{-i}=\pi_{-i}\left(s'_{-i},s_{i}\right)$. 
Consider a deviation by player $-i$ to the clustering
$f_{-i}^{\geq\pi'_{-i}}$ and
the following improvement path in
%Y(17.10): corrected "f_2" to "f_i", as suggested by reviewer 1.
$G_{\left(f_{-i}^{\geq\pi'_{-i}},f_{i}\right)}$
with respect to $s$. First, player $-i$ deviates to $s'_{-i}$ (which
strictly increases her clustered payoff). Next,
if $s_{i}$ is not a clustered best reply to $s'_{-i}$, then player $i$ changes
her strategy to a strategy $s'_{i}$ that is a clustered best
reply to $s'_{-i}$, and otherwise $s'_{i}=s_{i}$. The assumption that
$s_{i}\prec_{-i}BR_{i}\left(s'_{-i}\right)$ 
implies that $s_{i}\preceq_{-i}s'_{i}$. Observe that following these
two stages, the improvement path converges to a plausible equilibrium.
Since the game has monotone externalities, this plausible equilibrium
yields player $-i$ a strictly higher unclustered payoff than
$\pi_{-i}\left(s\right),$ which 
contradicts $s$ being a CUE outcome.

Part (2): Assume to the contrary that there exists a strategy profile
$s'$ satisfying $s'_{i}\preceq_{-i}s_{i}$, $s'_{-i}\preceq_{i}s_{-i}$,
$\pi_{i}\left(s'\right)>\pi_{i}\left(s\right)$, and either (a) $\pi_{-i}\left(s'\right)\geq\pi_{-i}\left(s\right)$
or (b) $s'_{-i}\in BR_{-i}\left(s_{i}'\right)$. Observe that monotone
externalities imply that $s'_{i}\neq
s_{i}$. Let
$\pi'_{i}=\pi_{i}\left(s'_{i},s_{-i}\right)$. 
Consider a deviation by player $i$ to the clustering
$f_{i}^{\geq\pi'_{i}}$ and 
the following improvement path in $G_{(f_{i}^{\geq\pi'_{i}},f_{-i})}$
with respect to $s$. First, player $i$ deviates to the strategy $s'_{i}$
that gives her a strictly higher clustered payoff
$$f_{i}^{\geq\pi'_{i}}\left(\pi_{i}\left(s'_{i},s_{-i}\right)\right)=\pi_{i}\left(s'_{i},s_{-i}\right)>\pi_{i}\left(s'\right)>\pi_{i}\left(s\right)=f_{i}^{\geq\pi'_{i}}\left(\pi_{i}\left(s\right)\right),$$ 
where the first inequality is due to the monotone externalities.
If $s_{-i}$ is a clustered best reply against $s'_{i},$ then $\left(s'_{i},s_{-i}\right)\in PNE(G_{(f_{i}^{\geq \pi'_{i}},f_{-i})},s)$
is a plausible equilibrium that gives the deviating player $i$ a
higher payoff, and we get a contradiction to $\left(s,f\right)$
being a CUE. Otherwise, player $-i$ deviates to $s'_{-i}.$ 

There are now two cases:
\begin{enumerate}
\item $\pi_{-i}\left(s'\right)\geq\pi_{-i}\left(s\right)$: Observe
that 
\begin{equation}
f_{-i}\left(BRP_{-i}\left(s_{i}\right)\right)\geq f_{-i}\left(BRP_{-i}\left(s'_{i}\right)\right)\geq f_{-i}\left(\pi_{-i}\left(s'\right)\right)\geq f_{-i}\left(\pi_{-i}\left(s\right)\right),
\label{eq-chain-ineq}
\end{equation}
where the first inequality is due to monotone externalities and the
last inequality is implied by $\pi_{-i}\left(s'\right)\geq\pi_{-i}\left(s\right)$.
The fact that $s \in NE (G_f)$ implies that 
$f_{-i}\left(BRP_{-i}\left(s_{i}\right)\right)=f_{-i}\left(\pi_{-i}\left(s\right)\right)$,
and thus all the terms in (\ref{eq-chain-ineq}) are equal to each other, which implies that  $s'\in PNE(G_{(f_{i}^{\geq\overline{\pi}_{i}},f_{-i})},s)$ (because $s'_{-i}$ is a clustered best reply to $s'_{i}$).

\item $s'_{-i}\in BR_{-i}\left(s_{i}'\right)$: It is immediate
that $s'\in PNE(G_{(f_{i}^{\geq \pi'_{i}},f_{-i})},s))$.
\end{enumerate}
In both cases, $s'$ is a plausible equilibrium that gives the
deviating player $i$ a higher payoff, so we get a contradiction to 
$\left(s,f\right)$ being a CUE.

\subsection{Proof of Theorem \ref{thm-complements}\label{proof-complements} (Strategic Complements)}

In order to prove Theorem \ref{thm-complements}, we need the following lemma:

\begin{lem}
  \label{Lemma-complements-standard} If $G$ is an interval game with
  strategic complements and monotone externalities, $s^{WNE}$ is 
the worst Nash equilibrium of $G$, and $s$ is a strategy profile satisfying
$s_{i}\succeq_{-i}BR_{i}\left(s_{-i}\right)$ for each player $i$,
then $s_{i}\succeq_{-i}s_{i}^{WNE}$ for each player $i$.
\end{lem}
\begin{proof}
  Assume to the contrary that there exists a player $j$ for which
  $s_{j}\prec_{-j}s_{j}^{WNE}$. 
Consider an auxiliary game $G^{R}$ similar to $G$ except that each
player $i$ is restricted to choosing a strategy $s_{i}$ satisfying
$s_{i}\preceq_{-i}s_{i}^{*}$. By a standard fixed-point theorem (\citealp{kakutani1941generalization}), the
restricted game admits a Nash equilibrium that we denote
$s^{RE}$. The strategy profile $s^{RE}$ cannot be a Nash 
equilibrium of $G$ because $s_{j}\succeq_{-i}s_{j}^{RE}$,
while $s_{j}\prec s_{j}^{WNE}$. This implies that there exists a player
$i$ for which $s_{i}^{RE}=s_{i}$ and
$s_{i}\prec_{-j}BR_{i}\left(s_{-i}^{RE}\right)$, 
which contradicts $s_{i}\succeq_{-i}BR_{i}\left(s_{-i}\right)\succeq_{-i}BR_{i}\left(s_{-i}^{RE}\right)$
(where the latter inequality is implied by  the assumption that $G$
has strategic complemenets and the fact that
$s_{-i}\preceq_{-i}s_{-i}^{RE}$).
\end{proof}

We can now prove Theorem \ref{thm-complements}. For the `if'' direction, suppose that Conditions
  (1--2) holds. Let $f^{\geq\pi\left(s\right)}=\left(f_{i}^{\geq\pi_{i}\left(s\right)},f_{-i}^{\geq\pi_{-i}\left(s\right)}\right)$
be the profile in which each player clusters all the payoffs above
her payoff in profile $s$. We show that $\left(f^{\geq\pi\left(s\right)},s\right)$
is a CUE. Assume to the contrary that $\left(f^{\geq\pi\left(s\right)},s\right)$
is not a CUE. Then there exists a player $i$, a clustering $f'_{i}$,
a plausible equilibrium $s'\in PNE(G_{\left(f'_{i},f_{-i}^{s}\right)})$
such that $\pi_{i}\left(s'\right)>\pi_{i}\left(s\right)$. Consider
an improvement path that converges to $s'$. The fact that $s_{j}\succeq_{j}BR_{j}(s_{-j})$
for each player $j$ implies that the first deviation of player $i$
is to an externalities-lower strategy with a strictly higher payoff,
that is, $s_{i}^{1}\prec_{-i}s_{i}$ and
$\pi_{i}\left(s_{i}^{1},s_{-i}\right)>\pi_{i}\left(s\right)$. 
Since the game has strategic complements, any payoff-improving deviation
in the second stage of any player $j$ must be to a strategy that
is externalities-lower $s_{j}$, that is, $s_{j}^{2}\preceq_{-j}s_{j}$. The same argument implies that at every later stage,
a payoff-improving deviation by any player $j$ must be to a strategy
$s_{j}^{k}\preceq_{-j}s_{j}$. Thus, the convergence point of the
improvement path $s'$ must be externalities-lower than $s$. The
fact that player $-i$ has clustering $f_{-i}^{\geq\pi\left(s\right)}$
implies that player $-i$ either 
\begin{enumerate}
\item obtains a payoff weakly higher than in $s$ (i.e.,
  $\pi_{-i}\left(s'\right)\geq\pi_{-i}\left(s\right)$, 
  which implies that $s'$ Pareto dominates $s$, violating condition (2a)), or
\item she plays an unclustered best reply ($s'_{-i}\in BR_{-i}\left(s'_{i}\right)$),
which violates condition (2b).
\end{enumerate}
Thus, both cases lead to a contradiction, which proves that $\left(f^{\geq\pi\left(s\right)},s\right)$
must be a CUE. 

For the ``only if'' direction, by Theorem
\ref{Thm:neccesary-conditions-monotone},
it suffices to consider the case where one of the players,
say player $2$, plays her unclustered best reply (i.e., $s_{2}\in
BR_{2}\left(s_{1}\right)$). 
Let $\left(f,s\right)$ be a CUE. We begin by showing that Condition
(1) holds. Assume to the contrary that player $1$ plays an externalities-lower reply, that is, $s_{1}\prec_{2}BR_{1}\left(s_{2}\right)$.
Let $s_{1}\prec_{2} s'_{1}\in BR_{1}\left(s_{2}\right)$.
Let $\pi'_{1}=\pi_{1}\left(s'_{1},s_{2}\right)>\pi_{1}\left(s\right)$.
Consider a deviation by player 1 to $f_{1}^{\geq\pi'_{1}}$, which
clusters together payoffs larger than $\pi'_{1}$, and the following
two-stage improvement path in $G_{\left(f_{1}^{\geq\pi'_{1}},f_{2}\right)}$
with respect to $s$:
\begin{enumerate}
\item Player 1 changes her strategy to $s'_{1}$ (which strictly increases
her clustered payoff). 
\item If $s_{2}$ is not a clustered best reply against $s'_{1}$, then
  player 2 changes her strategy to $s'_{2}\in BR_{2}\left(s'_{1}\right)$
  (observe that $s_{2}\preceq_{1}s'_{2}$ due to the game having strategic
complements, which further increases player 1's payoff).\\
\end{enumerate}
At the end of these two stages we have reached a plausible equilibrium
$\left(s'_{1},s'_{2}\right)\in PNE(G_{(f_{1}^{\geq\pi'_{1}},f_{2})},s)$
with a strictly higher payoff for player 1 (i.e.,
$\pi_{1}\left(s'_{1},s''_{2}\right)>\pi_{1}\left(s\right)$, 
which contradicts $\left(f,s\right)$ being a CUE.
The proof that condition (2) holds is essentially the same as
in Theorem \ref{Thm:neccesary-conditions-monotone}, and is omitted
for brevity.

Finally, we prove the ``moreover'' condition in the last sentence of
the theorem statement.  The inequality
$s_{i}\succeq_{-i}BR_{i}(s_{-i})$ 
implies that $s_{i}\succeq_{-i}s_{i}^{WNE}$ for each player $i$
by Lemma \ref{Lemma-complements-standard}). Finally, 
we have to show that $\pi_{i}\left(s\right)\geq\pi_{i}\left(s^{WNE}\right)$
for each player $i$. Assume to the contrary that one of the players,
say player $1$, obtains a strictly lower payoff than in the lowest
Nash equilibrium, that is, $\pi_{1}\left(s\right)<\pi_{1}\left(s^{WNE}\right).$
Consider a deviation by player 1 to $f_1^\emptyset$ (not clustering any
payoffs together). Consider the following improvement path. Let $s^{0}=s$.
In stage 1, if $s_{1}^{0}\notin BR_{1}\left(s_{2}^{0}\right)$, then
player 1 decreases her payoff to an unclustered best reply strategy $s_{1}^{1}\in BR_{1}\left(s_{2}^{0}\right),$ which satisfies $s_1^{WNE}\preceq s_{1}^{1}.$
Since $s_{1}^{1}\preceq s_{1}^{0}$ and the game has strategic complements,
$s_{2}^{1}=s_{2}^{0}\succeq_{1}BR_{2}(s_1^{1})\,$.
In stage 2, if $s_{2}^{1}\notin BR_{2}\left(s_{1}^{1}\right)$, then
player 2 decreases her strategy to an unclustered best reply $s_{2}^{2}\in
BR_{2}\left(s_{1}^{0}\right),$ and because the game has the strategic
complements, 
$s_{2}^{2}\succeq_{1}s_{2}^{WNE}.$ 
A straightforward induction show that (1) for every even $k$, in stage
$k+1$, if $s_{1}^{k}\notin 
BR_{1}\left(s_{2}^{k}\right)$, 
then player 1 decreases her strategy to an unclustered best reply, that is,
$s_{1}^{k+1}\in BR_{1}\left(s_{2}^{k}\right)$, which satisfies
$s_{2}^{k}\succeq_{2}s_{1}^{WNE}$ (because the game has strategic
complements), and (2)
for every odd $k$, in stage $k+1$, if $s_{2}^{k}\notin 
BR_{2}\left(s_{1}^{k}\right)$, 
then player 2 decreases her payoff to an unclustered best reply, that is,
$s_{2}^{k+1} \in BR_{1}\left(s_{2}^{k}\right)$, which satisfies $s_{2}^{k+1}\succeq_{1}s_{2}^{WNE}$
(the change must be a decrease for both even and odd $k$-s, since
the game has strategic complements). The fact that the players
always decrease their strategies (whenever they change them) implies
that the improvement path converges, and the limit $s'$ must be a
plausible equilibrium that satisfies (1) $s'_{1}\in
BR_{1}\left(s'_{2}\right)$, 
and (2) $s'_{i}\succeq_{-i}s_{i}^{WNE}$ for each player $i$. This
implies that $\pi_{1}\left(s'\right)\geq \pi_{1}(s_1^{WNE},s'_2)\geq \pi_{1}\left(s^{WNE}\right)>\pi_{1}\left(s\right)$,
which contradicts $\left(f,s\right)$ being a CUE.

\subsubsection{\label{proof-substitutes}Proof of Theorem \ref{thm-substitues} (Strategic Substitutes)}

\begin{enumerate}
    \item Assume to the contrary that $s_i\succ BR_i(s_{-i}).$ Let
      $s'_i \in BR_i (s_{-i})$ be an unclustered best reply
     strategy. Observe that  $s'_i  \prec s_i$. Let
     $\pi'_i= \pi_i(s'_i,s_{-i}).$ Consider a deviation of player $i$
     to the clustering $f_i^{\geq \pi'_i}.$ Consider the following
      improvement path. In the first stage, player $i$ changes her
      strategy to $s^1_1=s'_i.$ If $s_2$ is a clustered best reply of
      player 2, then this ends the improvement path. Otherwise, the
      improvement path includes an additional final stage in which
            player 2 changes her strategy to an unclustered best reply, i.e.,
      $s_{-i}^2\in BR_{-i}(s'_1).$ Strategic substitutability implies
      that $s_{-i}^2\succ s_{-i}.$ Observe that player 1 obtains a
      strictly higher payoff in the plausible Nash equilibrium that
      ends this improvement path relative to $\pi_i(s),$ which
      contradicts $(f,s)$ being a CUE. 
    \item Consider an auxiliary game $G^{R}$ similar to $G$ except
      that player $i$ is restricted to choosing strategies that are
      weakly externalities-lower than $s^i$. By a standard fixed-point
      theorem (\citealp{kakutani1941generalization}), the restricted game admits a Nash equilibrium that we
      denote $s^{RE}$. If $s^{RE}$ is not a Nash equilibrium of the
      original underlying game $G,$ then it must be that
      $s^{RE}_i=s_i$ and $s_i \prec BR_{-i}(s^{RE}_{-i}).$ The
      assumption that the payoff function is strictly convex implies
      that  $s_{-i}=s^{RE}_{-i}$ is the unique best reply to
      $s_{-i}$. Since the game has strategic substitutes,
      $s_i \succeq BR_{-i}(s^{RE}_{-i})$, so we get a a
      contradiction. Thus, $s^{RE}$ must be a Nash equilibrium of the
      unrestricted game $G$. It is immediate that $s_i\preceq_i
      s^{RE}_i.$ Finally, the fact that the game has strategic
      substitutes implies that $s_{-i}\succeq_i s^{RE}_{-i}.$ 
\end{enumerate}

\section{Partial Observability}\label{sec-partial}
Throughout the paper we assume that if an agent deviates to a different clustering, then the opponent always observes this deviation. In this appendix, we relax this assumption, and show that
our results also hold in a setup with partial observability (most
results hold for any level of partial 
observability, while the remaining results hold for a sufficiently
high level of observability).  Our partial-observability extension is analogous to that of \citet[Online appendix E]{heller2020biased}, and in the
spirit of the observation structure of
\citet{DekelElyEtAl2007Evolution}. 

\subsection {Adapted Model}
Let $p\in [0,1]$ denote the probability that an agent who is matched
with an opponent who deviates to a different clustering
\emph{privately} observes the opponent's deviation (henceforth,
\emph{observation 
probability}). If an agent does not observe the deviation, then she continues playing her original CUE strategy.

We define a $p$-restricted coarse-utility game as a game between an
incumbent (player $-i$) and a deviator (player $i$) in which the
incumbent is restricted to playing her original strategy $s$  with
probability 
%Y(12.4): I coorected a mistake of having hhere %p% instead of $1-p$
$1-p$ (i.e., when not observing the opponent's deviation).  Formally:

\begin{defn}
Fix player $-i$, clustering profile $f$ and strategy $s_{-i}$. The payoff function $\pi^p_j$  of each player $j$ in the \emph{$p$-restricted coarse-utility game} $G^p_{f,s_{-i}}=(S,f \circ \pi^p)$ is defined as follows:
 $$\pi_j^p(s')= p\cdot f_j\circ \pi_{j}(s') + (1-p) \cdot f_j \circ \pi_{j}(s_i,s_{-i}').$$
\end{defn}
The deviator (player $i$) is aware that her different clustering is
privately observed by her opponent (player $-i$) with probability
$p$. Thus, the deviator faces two different possible payoffs (one when
her clustering is observed by her opponent, and one when it is not
observed); she evaluates each payoff using her coarse utility
($f_j\circ \pi_{j}(s')$ and $f_j \circ \pi_{j}$, respectively).
She next evaluates her expected payoff as the mixed average of
these two outcomes ($\pi_j^p(s')$), and uses  her coarse utility
to obtain her final evaluation of her payoff ($f_j \circ \pi_j (s')$). Observe that in this appendix (unlike in the model used in the
main text), the clustered payoff function has cardinal meaning (as it is
used in the expected payoff calculation of $\pi_{j}$). 

The set $NE(G^p_{f,s_{-i}})$ of Nash equilibria of $G^p_{f,s_{-i}}$ is defined in
the standard way. Next, we adapt our definition of plausible equilibrium.
\begin{defn}
\label{def:plausible-p}Fix a strategy profile $s$ and a clustering
profile $f$. An equilibrium $s'\in NE(G^p_{f,s_{-i}})$ is \emph{$p$-plausible}
with respect to $s$ if there is a sequence $\left(s^{k}\right)_{k\geq0}$
of strategy profiles satisfying: (1) $s^{0}=s$, (2) $\lim_{k\rightarrow\infty}s^{k}=s'$,
and (3) if $s_{i}^{k+1}\neq s_{i}^{k}$ for player $j$ and $k\geq0$,
then $f_j \circ \pi_{j}^p(s_{i}^{k+1},s_{-i}^{k})> f_j \circ
\pi_{j}^p(s_{i}^{k},s_{-i}^{k})$. 
\end{defn}
Let $PNE^p\left(G_{f},s\right)$ be the set of $p$-plausible equilibria
with respect to $s$. Next, we adapt the three variants of our solution concept to the setup of partial observability. 
\begin{defn}
    A pair $(f,s)$, where $f$ is a clustering profile and $s$ is a
strategy profile satisfying $f\in NE (G_f)$, is a 
\begin{enumerate}
    \item \emph {weak $p$-CUE} if for each player $i$ and each
            clustering $f'_{i}$, there 
exists an equilibrium $s'\in NE(G^p_{(f'_i,f_{-i}),s_{-i}})$ 
such that $\pi_{i}\left(s'\right)\leq\pi_{i}\left(s\right)$. 
    \item  \emph {$p$-CUE} if $\pi_{i}\left(s'\right)\leq\pi_{i}\left(s\right)$ for each player $i$, each clustering $f'_{i}$, and each $p$-plausible equilibrium $s'\in PNE(G^p_{(f'_i,f_{-i}),s_{-i}}).$ 
    \item \emph {strong $p$-CUE} if $\pi_{i}\left(s'\right)\leq\pi_{i}\left(s\right)$ for each player $i$, each clustering $f'_{i}$, and each equilibrium $s'\in NE(G^p_{(f'_i,f_{-i}),s_{-i}}).$ 
\end{enumerate}
\end{defn}
\subsection{Adapted Results}
We begin by adapting Proposition \ref{prop:strong-necce} to the current setup.
\begin{prop}
Fix $p\in[0,1]$. Let $\left(f,s\right)$ be a strong $p$-CUE,
and let $s^{NE}$ be a Nash equilibrium. Then $\pi_{i}\left(s\right)\geq p\cdot \pi_{i}\left(s^{NE}\right)+(1-p)\cdot \pi_{i}\left(s_i,s^{NE}_{-i}\right))$
for all players $i$.
\end{prop}
\begin{proof} 
The proof is immediate by considering a deviation of player $i$ to
clustering all payoffs together (i.e., $f_i^{\mathbb{R}}$) and the fact
that $s^{NE}\in NE(G^p_{(f'_i,f_{-i}),s_{-i}})$. 
\end{proof}

It is easy to see that Proposition \ref{fact2:any-NE-is-NE-clustered-1}
can be extended to our setup.  
\begin{prop}
Let $s^{NE}$ be a Nash equilibrium
of the underlying game. Then $\left(f,s^{NE}\right)$ is a $p$-CUE for every clustering profile $f$ and every $p\in[0,1]$.
\end{prop}
\begin{proof}
The proposition holds because $s^{NE}\in PNE^p\left(G_{\left(f'_{i},f_{-i}\right)},s\right)$
for every player $i$ and clustering $f'_{i}$ with respect to the
constant improvement path
$\left(s^{k}\right)_{k\geq0}$=$\left(s^{NE},s^{NE},...\right)$.  
\end{proof}
Next we show that the ``folk-theorem'' result for weak $p$-CUE holds
for $p$ sufficiently high (extending Proposition
\ref{prop-falk-thm}). 
\begin{prop}
Let $s$ be a strategy profile.
\begin{enumerate}
\item If $s$ is a weak $p$-CUE outcome for $p \in [0,1]$, then $s$
  is individually rational.  
  \item If $s$ is strongly individually rational, then for some
  $\overline p<1$,  $s$ is a weak CUE outcome for all $p \geq
  \overline p.$ 
\end{enumerate}
\end{prop}
The proof, which is analogous to the proof of Proposition \ref{prop-falk-thm}, is omitted for brevity.

{
In contrast to the ``folk-theorem'' result above, which requires the probability $p$ to be sufficiently high, if $p$ is sufficiently low, 
then only Nash equilibria can be $p$-CUE outcomes. That is, observability is necessary for coarse 
utility to introduce non-Nash behavior.}
\begin{prop} \label{prop-no-observability-only-Nash}{
Let $s$ be a strategy profile.}
    \begin{enumerate}
        \item {If $s$ is a weak 0-CUE outcome, then $s$ is a Nash equilibrium of the underlying game.}
        \item {if $s$ is not a Nash equilibrium, then for some $\bar p>0$, $s$ is 
not a weak CUE outcome for all $p<\bar p$.}
    \end{enumerate}
\end{prop}
\begin{proof}
    {The simple proof adapts an insight of} \citet{ok2001evolution} {from the evolution of preferences to the current setup.}
    \begin{enumerate}
    
        \item 
        {Assume to the contrary that $(f,s)$ is a weak CUE and $s$ is not a Nash equilibrium of the underlying game. This implies that there exists a player $i$ and a strategy $s'_i$ such that $\pi_i(s_i',s_{-i})>\pi_i(s).$ 
        Observe that following a deviation of player $i$ to not clustering payoffs (i.e., to ${f'_{i}}^{\emptyset}$), player $-i$ must continue playing $s_{-i}$ in any equilibrium of the post-deviation game (due to not observing the opponent's deviation), which implies that player $i$ obtains payoff 
        $\pi_i(s_i',s_{-i})>\pi_i(s)$ in all equilibria 
       $s'\in NE(G^0_{({f'_i}^\emptyset,f_{-i}),s_{-i}}),$
        which contradicts $s$ being a weak CUE outcome.
        }
        \item 
        {
        Because $s$ is not a Nash equilibrium of the 
        underlying game, there exists a player $i$ who 
        can gain $\delta$ by deviating. Let $\bar p$ be any positive number that is smaller than $\delta$ divided by the difference between the maximal and the minimal feasible payoff of player $i.$ Observe that $s$ cannot be a 
       weak $p$-CUE outcome for $p<\bar p$ because player $i$, by deviating to not clustering payoffs (i.e., to ${f'_{i}}^{\emptyset}$), gains at least
        $\delta$ when her deviation is not observed, which outweighs her maximal feasible expected loss when her deviation is observed.}\qedhere
  
    \end{enumerate}
\end{proof}
It is easy to see that the Proposition \ref{claim-constant-sum}
(characterization of $p$-CUE outcomes in constant-sum games),
Proposition \ref{claim-common-interest} (characterization of $p$-CUE
outcomes in games with common interests), and Theorem
\ref{Thm:neccesary-conditions-monotone} hold for all $p\in [0,1]$
(where $p$-CUE replaces CUE in the statement of each result) with
minor adaptations to the proofs. 

We next observe that the necessary conditions for being a $p$-CUE
outcome in games with strategic complements are the same as in Theorem
\ref{thm-complements}. By contrast, these conditions are no longer
sufficient for being $p$-CUE outcomes. Specifically, one can show that
lower values of $p$ have smaller sets of CUE outcomes, which converge
towards the set of Nash equilibria as $p$ converges to zero. Formally,
the adaptation of the necessary conditions of Theorem
\ref{thm-complements} to the current setup is as follows (the proof,
which is analogous to the proof of Theorem \ref{thm-complements}, is
omitted for brevity): 

\begin{prop}
If $G$ is an interval game with monotone externalities and strategic
complements, $s$ is an interior 
strategy profile, $p\in [0,1]$, and $s$ is a $p$-CUE outcome, then
\begin{enumerate}
\item $s_{i}\succeq_{-i}BR_{i}\left(s_{-i}\right)$
for each player $i$;
\item if $s'_{i}\preceq_{-i}s_{i}$, $s'_{-i}\preceq_{i}s_{-i}$, and
  either (a) $\pi_{-i}\left(s'\right)\geq\pi_{-i}\left(s\right)$ or
  (b) $s'_{-i}\in BR_{-i}\left(s_{i}'\right)$, then
  $\pi_{i}\left(s'\right)\leq\pi_{i}\left(s\right)$. 
\end{enumerate}
\end{prop}
Finally, minor adaptations to the proof of Theorem \ref{thm-substitues} show that the necessary conditions for being a $p$-CUE outcome in games with strategic substitutes are the same as in Theorem \ref{thm-substitues} (where $p$-CUE outcome replaces CUE outcome). 
\fullv{
\bibliographystyle{chicago}
\bibliography{CUE}
}
\end{document}